\documentclass[%
 aip,
 jmp,%
 amsmath,amssymb,
]{revtex4-1}

\usepackage{amsmath,amssymb}

\usepackage{amsthm}
\newtheorem{theo}{Theorem}
\newtheorem{lemm}{Lemma}
\newtheorem{prop}{Proposition}
\newtheorem{coro}{Corollary}

\theoremstyle{definition}
\newtheorem{defi}{Definition}
\newtheorem{rem}{Remark}
\newtheorem{exam}{Example}

\usepackage{enumerate}
\usepackage{braket}
\usepackage{bbm}
\usepackage{dsfont}

\usepackage{graphicx}
\usepackage{dcolumn}
\usepackage{bm}

\newcommand{\defarrow}{\stackrel{\mathrm{def.}}{\Leftrightarrow}}

\newcommand{\cmplx}{\mathbb{C}}
\newcommand{\realn}{\mathbb{R}}

\newcommand{\natn}{\mathbb{N}}

\newcommand{\tr}{\mathrm{tr}}

\newcommand{\cL}{\mathcal{L}}
\newcommand{\cH}{\mathcal{H}}
\newcommand{\Hin}{\mathcal{H}_\mathrm{in}  }
\newcommand{\Hout}{\mathcal{H}_\mathrm{out}  }
\newcommand{\cK}{\mathcal{K}}
\newcommand{\LH}{\mathcal{L} (\mathcal{H}) }
\newcommand{\LHin}{\mathcal{L} (\mathcal{H}_\mathrm{in}  ) }
\newcommand{\LHout}{\mathcal{L} (\mathcal{H}_\mathrm{out}  ) }
\newcommand{\LHone}{\mathcal{L}  (\mathcal{H}_1  )}

\newcommand{\LK}{\mathcal{L}  (\mathcal{K})}
\newcommand{\LKone}{\mathcal{L}  (\mathcal{K}_1  )}

\newcommand{\SH}{\mathcal{S} (\mathcal{H}) }
\newcommand{\SHin}{\mathcal{S} (\mathcal{H}_{\mathrm{in}}) }

\newcommand{\Th}{\mathcal{T} (\mathcal{H}) }

\newcommand{\cstar}{$C^\ast$}
\newcommand{\A}{\mathcal{A}}
\newcommand{\B}{\mathcal{B}}
\newcommand{\C}{\mathcal{C}}
\newcommand{\D}{\mathcal{D}}
\newcommand{\Ain}{\mathcal{A}_{\mathrm{in}}}

\newcommand{\M}{\mathcal{M}}
\newcommand{\N}{\mathcal{N}}
\newcommand{\Min}{\mathcal{M}_{\mathrm{in}}}

\newcommand{\chset}[2]{ \mathbf{Ch}  (#1 \to #2 )  }
\newcommand{\ch}[1]{\mathbf{Ch}  (#1 )}
\newcommand{\nchset}[2]{ \mathbf{Ch}_\sigma  (#1 \to #2 )  }
\newcommand{\nch}[1]{\mathbf{Ch}_\sigma  (#1 )}

\newcommand{\Thmin}{\Theta_{\mathrm{min}}}
\newcommand{\Thmax}{\Theta_{\mathrm{max}}}

\newcommand{\E}{\mathcal{E}}

\newcommand{\I}{\mathcal{I}}

\newcommand{\vph}{\varphi}

\newcommand{\cS}{\mathcal{S}}
\newcommand{\Ss}{\mathcal{S}_{\sigma}}

\newcommand{\oM}{\mathsf{M}}
\newcommand{\oN}{\mathsf{N}}
\newcommand{\oP}{\mathsf{P}}

\newcommand{\atensor}{\otimes_{\mathrm{alg}}}
\newcommand{\ntensor}{\overline{\otimes}}
\newcommand{\maxtensor}{\otimes_{\mathrm{max}}}
\newcommand{\mintensor}{\otimes_{\mathrm{min}}}
\newcommand{\gtensor}{\otimes_{\gamma}}
\newcommand{\btensor}{\otimes_{\mathrm{bin}}}

\newcommand{\id}{\mathrm{id}}
\newcommand{\1}{\mathds{1}}

\newcommand{\cocp}{\preccurlyeq_{\mathrm{CP}}}
\newcommand{\eqcp}{\sim_{\mathrm{CP}}}
\newcommand{\concp}{\preccurlyeq_{ \mathrm{CP}_\sigma}}
\newcommand{\eqncp}{\sim_{\mathrm{CP}_\sigma}}

\newcommand{\postp}{\preceq}
\newcommand{\posteq}{\simeq}

\newcommand{\comp}{\bowtie}
\newcommand{\ncomp}{\comp_{\sigma}}
\newcommand{\maxcomp}{\comp_{\mathrm{max}}}
\newcommand{\mincomp}{\comp_{\mathrm{min}}}
\newcommand{\bcomp}{\comp_{\mathrm{bin}}}
\newcommand{\gcomp}{\comp_{\gamma}}

\newcommand{\norm}[1]{\lVert #1 \rVert}
\newcommand{\maxnorm}[1]{\norm{ #1 }_{\mathrm{max}}}
\newcommand{\minnorm}[1]{\norm{ #1 }_{\mathrm{min}}}
\newcommand{\bnorm}[1]{\norm{ #1 }_{\mathrm{bin}}}
\newcommand{\gnorm}[1]{\norm{ #1 }_{\gamma}}

\newcommand{\phth}{\varphi_\theta}
\newcommand{\psth}{\psi_\theta}
\newcommand{\F}{\mathcal{F}}
\newcommand{\seE}{(\M , \Theta , (\varphi_\theta)_{\theta \in \Theta} )}
\newcommand{\seF}{(\N , \Theta , (\psi_\theta)_{\theta \in \Theta} )}
\newcommand{\thin}{\theta \in \Theta}

\begin{document}
\preprint{}
\title[]{%
Quantum incompatibility of channels with general outcome operator algebras%
}
\author{Yui Kuramochi}
\affiliation{School of Physics and Astronomy, Sun Yat-Sen University (Zhuhai Campus),
Zhuhai 519082, China}
\email{yui.tasuke.kuramochi@gmail.com}
\date{}

\begin{abstract}
A pair of quantum channels are said to be incompatible if they cannot be realized
as marginals of a single channel.
This paper addresses the general structure of 
the incompatibility of completely positive channels
with a fixed quantum input space and with general outcome operator algebras.
We define a compatibility relation for such channels by identifying 
the composite outcome space as the maximal (projective) $C^\ast$-tensor product
of outcome algebras.
We show theorems that characterize this compatibility relation
in terms of the concatenation and conjugation of channels,
generalizing the recent result for channels with quantum outcome spaces.
These results are applied to the positive operator valued measures (POVMs)
by identifying each of them with the corresponding quantum-classical (QC) channel.
We also give a characterization of the maximality of a POVM 
with respect to the post-processing preorder
in terms of the conjugate channel of the QC channel.
We consider another definition of compatibility of normal channels 
by identifying the composite outcome space 
with the normal tensor product of the outcome von Neumann algebras.
We prove that for a given normal channel the class of 
normally compatible channels is upper bounded by
a special class of channels called tensor conjugate channels.
We show the inequivalence of the $C^\ast$- and normal 
compatibility relations for QC channels,
which originates from the possibility and impossibility of copying operations
for commutative von Neumann algebras in $C^\ast$- and normal 
compatibility relations, respectively.
\end{abstract}
\maketitle

\section{Introduction} \label{sec:introduction}
Impossibility of simultaneous realizations of two quantum operations 
is a fundamental feature of quantum theory.
For example,
we cannot simultaneously measure the position and momentum of 
a quantum particle due to the celebrated uncertainty relation.
In general, a pair of quantum operations are said to be compatible 
if they can be realized as marginals of a single operation 
and incompatible if cannot.
Recently
the concept of incompatibility of 
completely positive (CP) quantum channels with input and outcome quantum spaces
has been introduced\cite{1751-8121-49-12-123001,1751-8121-50-13-135302},
which is applicable to the measurements on a pair of discrete observables.
In Ref.~\onlinecite{1751-8121-50-13-135302}, it is proved 
that for a finite-dimensional quantum channel $\Lambda ,$
its conjugate (or complementary) 
channel~\cite{devetak2005capacity,holevo2007complementary,king2007properties} 
$\Lambda^c$
is maximal in the concatenation preorder among the channels compatible with $\Lambda .$
In addition, it is also shown that
the class of channels compatible with a given channel $\Lambda$
characterizes the equivalence class of $\Lambda$
with respect to the concatenation equivalence relation.

Then one may ask, as a natural generalization, 
whether we have similar results for channels with a quantum input space 
and with general outcome operator algebras,
especially for measurements on continuous observables like position and momentum.
This paper addresses this generalization problem and 
gives an affirmative answer.

Now we outline the contents of this paper.
For the generalization to arbitrary outcome operator algebras, 
one of the main difficulties is the proper choice of 
the tensor product of the outcome algebras corresponding to the composite outcome system.
As known in the field of operator algebras, 
a given pair of \cstar-algebras have in general many inequivalent \cstar-tensor products.
We can also define normal tensor product if the outcome operator algebras 
are von Neumann algebras. 
These notions of tensor products lead to various definitions of compatibility relations
of channels,
which are consistent with the (normal) concatenation preorder relation for channels
(Sections~\ref{sec:prel} and \ref{sec:channel}).

The main finding of this paper is that,
if we define the compatibility relation by the maximal (or projective) tensor product
(Definition~\ref{defi:chcomp})
and generalize the concept of the conjugate channel to the commutant conjugate channel
(Definition~\ref{defi:cconj}),
then, 
for channels with an input quantum space and general outcome \cstar-algebras,
the commutant conjugate channel is maximal 
in the concatenation preorder
among the compatible channels
(Section~\ref{sec:main}, Theorems~\ref{theo:univbmax} and \ref{theo:univmax}).
We also establish that the class of compatible channels determines the concatenation 
equivalence class of a normal channel (Theorem~\ref{theo:equiv}).
As a by-product of the proof of this,
we find that the normal and non-normal concatenation relations coincide
for normal channels 
(Corollary~\ref{coro:nnonn})
and for statistical experiments (Corollary~\ref{coro:se}).

We also consider the quantum-classical (QC) channel,
which is a normal channel with a commutative outcome von Neumann algebra
and corresponds to a quantum measurement
described by a positive operator valued measure (POVM)
(Section~\ref{sec:povm}).
In this case, since \cstar-tensor product is unique if 
one of the algebras is commutative,
all the compatibility relations defined by \cstar-tensor products coincide.
We define the compatibility of a POVM and a channel by 
the existence of a joint instrument of them 
and show that this compatibility is consistent with the compatibility of the QC channel
(Theorem~\ref{theo:povmch}).
We also give a characterization of the maximality~\cite{Dorofeev1997349,ISI:000236487200002} 
of a POVM
with respect to the classical post-processing preorder
in terms of the conjugate channel of the QC channel.

Another natural choice of tensor product for normal channels
is the normal ($W^\ast$-) tensor product of outcome von Neumann algebras.
We show that, for a given normal channel,
the class of normally compatible channels is upper bounded in the concatenation preorder
by a special class of channels with quantum outcome spaces
called tensor-conjugate channels (Section~\ref{sec:normal}).
The tensor conjugate channels of a given normal channel are, in general, 
not concatenation equivalent to the commutant conjugate channel.
We also compare the \cstar- and normal compatibility relations focusing on
QC channels (Section~\ref{sec:comparison}).
We show that, for a commutative von Neumann algebra, 
the normal compatibility relation does not allow 
the universal copying (or broadcasting) operation 
unless the algebra
is atomic, whereas the \cstar-compatibility relation does.
This fact leads to an example of a POVM that is normally incompatible with itself,
which explicitly shows the inequivalence of the \cstar- and normal compatibility relations.
Finally we give some open problems related to our results
(Section~\ref{sec:conclusion}).

\section{Preliminaries on operator algebras and channels} \label{sec:prel}
In this section we introduce some preliminaries 
on the operator algebras and CP channels between them.
For general references we refer to 
Refs.~\onlinecite{takesakivol1,brown2008c,paulsen_2003}.

\subsection{\cstar-algebra}
Let $\A$ be a \cstar-algebra.
Throughout this paper we only consider unital \cstar-algebras
and the unit element of $\A$ is denoted by $\1_{\A} .$
A positive linear functional $\vph \in \A^\ast$ is called a state 
if $\vph (\1_{\A}) =1$ and the set of states on $\A$ is denoted
by $\cS (\A) .$
For each linear functional $\varphi \in \A^\ast$ and $A \in \A ,$
$\varphi (A)$ is also written as $\braket{ \varphi   ,  A  } .$

Let $\A$ and $\B$ be \cstar-algebras
and $\Lambda \colon \A \to \B$ be a linear map.
$\Lambda$ is called unital if $\Lambda (\1_\A) = \1_\B .$
$\Lambda$ is called positive if 
$\Lambda (A) \geq 0$
for $0 \leq A \in \A . $
For positive $\Lambda ,$
$\Lambda$ is called faithful if $\Lambda (A^\ast A) = 0$ implies
$A = 0$ for $A \in \A .$
$\Lambda $ is called CP
if 
\[
	\sum_{i,j =1}^n
	B^\ast_i 
	\Lambda (A^\ast_i A_j)
	B_j
	\geq 0
\]
for each $n \geq 1 , $
$A_i \in \A ,$
and $B_i \in \B $
($ 1\leq i \leq n$).
$\Lambda$ is called a channel (in the Heisenberg picture) 
if $\Lambda$ is unital and CP.
The set of channels from $\A$ to $\B$ is denoted by
$\chset{\A}{\B}$
and $\chset{\A}{\A}$ is written as $\ch{\A} .$
For each channel $\Gamma \in \chset{\A}{\B} ,$
the codomain $\B$ and the domain $\A$
are called the input and outcome spaces of $\Gamma ,$
respectively.
The identity channel on $\A$ is denoted by $\id_\A .$

Let $\cH$ be a Hilbert space.
The inner product on $\cH$ is denoted by 
$\braket{\psi|\vph}$
($\psi , \vph \in \cH$),
which is linear and antilinear with respect to 
$\vph$ and $\psi ,$ respectively.
We denote by $\LH ,$ $\Th ,$ and $\SH$ 
the sets of bounded, trace-class, and density operators on 
$\cH ,$ respectively.
The identity operator on $\cH$ is written as $\1_{\cH} .$

Let $\A$ be a \cstar-algebra,
let $\cH$ be a Hilbert space,
and let $\Lambda \in \chset{\A}{\LH}$ be a channel.
A triple $(\cK , \pi  , V)$ is called a Stinespring representation of $\Lambda$
if
$\cK$ is a Hilbert space,
$\pi \colon \A \to \LK$ is a representation 
(i.e.\ non-degenerate $\ast$-homomorphism)
from $\A$ to $\cK ,$
$V \colon \cH \to \cK$ is a linear isometry,
and $\Lambda$ is represented as
\[
	\Lambda (A) = V^\ast \pi (A) V
	,
	\quad
	(\forall A \in \A) .
\]
A Stinespring representation $(\cK , \pi  , V)$ of $\Lambda$
is called minimal if 
the closed linear span of
$\pi (\A) V \cH $
coincides with $\cK .$
According to Stinespring's dilation theorem~\cite{1955stinespring},
any channel $\Lambda \in \chset{\A}{\LH}$ has a 
minimal Stinespring representation
unique up to unitary equivalence.

For a channel 
$\Lambda \in \chset{\A}{\B} ,$
the set
\[
	M_\Lambda
	:=
	\Set{
	A  \in \A
	|
	\Lambda(A^\ast A)
	=
	\Lambda (A^\ast) \Lambda (A) , \,
	\Lambda (A A^\ast)
	=
	\Lambda (A) \Lambda (A^\ast) 
	}
\]
is called the multiplicative domain of $\Lambda .$
For $A \in \A ,$
$A \in M_\Lambda$
if and only if
$\Lambda (BA) = \Lambda (B) \Lambda (A) $
and
$\Lambda (AB) = \Lambda (A) \Lambda (B) $
for all $B \in \A .$

Let $\A$ be a \cstar-algebra.
The (universal) enveloping von Neumann algebra $\A^{\ast \ast }$
of $\A$ is the von Neumann algebra 
defined by $\pi_U (\A)^{\prime \prime} ,$
where $(\pi_U , \cK_U)$ is the direct sum of GNS-representations
taken over all the states on $\A .$
Here, 
for a Hilbert space $\cH $ and
a subset $S \subseteq \LH ,$
$S^\prime$ denotes the commutant of $S.$
By the faithfulness of the representation $\pi_U ,$
$\A$ can be regarded as a \cstar-subalgebra of $\A^{\ast \ast}$
ultraweakly dense in $\A^{\ast \ast} .$
The enveloping algebra $\A^{\ast \ast}$ satisfies the following universal property:
for each Hilbert space $\cK$ and each representation 
$\pi \colon \A \to \LK ,$ 
$\pi$ extends to a unique normal representation
$\tilde{\pi} \colon \A^{\ast \ast } \to \LK .$
The enveloping von Neumann algebra $\A^{\ast \ast}$ is, 
as a Banach space, 
isometrically isomorphic to the double dual of $\A ,$
which justifies the notation $\A^{\ast \ast } .$

\subsection{Von Neumann algebra}
Let $\M$ and $\N$ be von Neumann algebras.
A positive linear map
$\Lambda \colon \M \to \N$
is called normal if 
$\Lambda (\sup_\alpha A_\alpha)   =  \sup_\alpha \Lambda (A_\alpha)$
holds for any upper and lower bounded monotonically increasing net $A_\alpha$
on $\M .$
A positive linear map 
$\Lambda \colon \M \to \N$
is normal if and only if it is ultraweakly continuous.
A normal, unital, CP map $\Lambda \colon \M \to \N$
is called a normal channel.
The set of normal channels from $\M$ to $\N$
is denoted by
$\nchset{\M}{\N} $
and $\nchset{\M}{\M}$ by $\nch{\M} .$
If 
$(\cK , \pi , V)$ is a minimal Stinespring representation
of a normal channel $\Lambda \in \nchset{ \M  }{ \LH} ,$
$\pi$ is a normal representation of $\M .$

Let $\M$ be a von Neumann algebra.
The set of ultraweakly continuous elements of the dual space $\M^\ast$ is denoted by
$\M_\ast .$
$(\M_\ast)^\ast$ is identified with $\M$ by the correspondence
$\M \ni A \mapsto \Psi_A \in (\M_\ast)^\ast ,$
where $\Psi_A (\varphi) := \braket{ \varphi , A   },$
$(\varphi \in \M_\ast) .$
In this sense, $\M_\ast$ is the predual space of $\M .$
The set of normal states on $\M$ is denoted by $\Ss (\M) .$

Let $\M$ and $\N$ be von Neumann algebras and let $\Lambda \in \nchset{\N}{\M}$
be a normal channel.
Then there exists a unique positive linear map  
$\Lambda_\ast \colon \M_\ast \to \N_\ast$
such that
\begin{equation*}
	\braket{ \varphi , \Lambda(A)}
	=
	\braket{ \Lambda_\ast (\varphi)  ,  A}
	\quad
	(A \in \N , \varphi \in \M_\ast),
\end{equation*}
and $\Lambda_\ast$ is called the predual of $\Lambda .$
In the physical context, the predual of a channel corresponds to the Schr\"odinger picture.

For a Hilbert space $\cH ,$ $\LH$ is called a fully quantum space and
a channel with fully quantum input and outcome spaces is called a 
fully quantum channel.
The predual 
$\LH_\ast$ (respectively, the set of normal states $\Ss (\LH)$) is identified with
$\Th$ (respectively, $\SH$) 
by the correspondence
$\braket{T , A} = \tr [T A] ,$
$(T \in \Th , A \in \LH) .$

\subsection{Tensor products of operator algebras}
Let $\A$ and $\B$ be \cstar-algebras
and let $\A \atensor \B$ denote the algebraic tensor product of $\A$ and $\B .$
A \cstar-norm $\gnorm{\cdot}$ on $\A \atensor \B$
is a norm such that
$\gnorm{ XY }  \leq \gnorm{ X }  \gnorm{ Y } $
and 
$\gnorm{ X^\ast X } = \gnorm{ X }^2 $
for all $X , Y \in \A \atensor \B .$
The completion of $\A \atensor \B$ 
with respect to a \cstar-norm $\gnorm{\cdot} ,$
denoted by $\A \gtensor \B ,$
is a \cstar-algebra and called a \cstar-tensor product of $\A$ and $\B .$
For general $\A$ and $\B ,$
\cstar-norm is not unique.
The following minimal and maximal norms are important:
for each $X \in \A \atensor \B $ we define 
\begin{gather*}
	\minnorm{X}
	:=
	\sup 
	\set{ \norm{ \pi_1 \otimes \pi_2 ( X) }   | 
	\text{$\pi_1$ and $\pi_2$ are representations of $\A$ and $\B ,$ respectively}
	} ,
	\\
	\maxnorm{X}
	:=
	\sup 
	\set{ \norm{ \pi ( X) }   | 
	\text{$\pi$ is a representation of $\A \atensor \B $}
	} ,
\end{gather*}
Then $\minnorm{\cdot}$ and $\maxnorm{\cdot}$ are 
\cstar-norms.
The \cstar-algebras $\A \mintensor \B $ and $\A \maxtensor \B$
are called the minimal (or injective) and maximal (or projective) tensor products of 
$\A$ and $\B$, respectively.
For each \cstar-norm 
$\gnorm{\cdot}$
on $\A \atensor \B ,$
we have 
$
	\minnorm{X} 
	\leq 
	\gnorm{X}
	\leq
	\maxnorm{X}
$
for all $X \in \A \atensor \B  .$
The maximal tensor product can be characterized by the following universal property:
for any representations
$
\pi_1 \colon \A  \to \cL  (\cH)
$
and 
$
\pi_2 \colon \B \to \cL  (\cH)
$
with commuting ranges, 
there exists a unique representation
$
\pi_{12} \colon \A \maxtensor \B  \to \cL  (\cH)
$
such that 
$\pi_{12} (A \otimes B) = \pi_1 (A) \pi_2 (B)$
$(A \in \A , B \in \B ).$
For any \cstar-tensor product
$\A \gtensor \B$ there exist homomorphisms
$\pi_{\max} \colon \A \maxtensor \B \to \A \gtensor \B $
and 
$\pi_{\min} \colon   \A \gtensor \B \to  \A \mintensor \B  $
such that $\pi_{\max} (A \otimes B ) = A \otimes B $ and
$\pi_{\min} (A \otimes B ) = A \otimes B $
$(A \in \A , B \in \B) .$

A \cstar-algebra $\A$ is called nuclear if \cstar-norm on 
$\A \atensor \B$ is unique for any \cstar-algebra $\B .$
Any commutative \cstar-algebra is nuclear.

Let $\M$ and $\N$ be von Neumann algebras on Hilbert spaces
$\cH$ and $\cK ,$ respectively.
We define a \cstar-norm $\bnorm{\cdot}$ on $\M \atensor \N$ by
$
	\bnorm{X}
	:=
	\sup_{\pi_{12}} 
	\norm{ \pi_{12} ( X) }  ,
$
where the supremum is taken over the representation
$\pi_{12}$ such that the marginals 
\begin{gather*}
	\pi_1 \colon
	\M \ni A \mapsto \pi_{12} (A \otimes \1_\N) 
	\\
	\pi_2 \colon
	\N \ni B \mapsto \pi_{12} (\1_\M \otimes B)
\end{gather*}
are normal~\cite{EFFROS19771}.
The corresponding \cstar-tensor product 
$\M \btensor \N $
is called the binormal tensor product
of $\M$ and $\N .$
We can also define the normal tensor product of $\M$ and $\N$ as follows.
We embed $\M \atensor \N$ into $\cL  ( \cH  \otimes \cK)$
and define the normal tensor product $\M \ntensor \N$
by 
$(\M \atensor \N)^{\prime \prime} ,$
the von Neumann algebra generated by $\M \atensor \N$
in $\cL  ( \cH  \otimes \cK) .$
The norm $\norm{\cdot}$ on $\M \ntensor \N$
restricted to $\M \atensor \N$ coincides with
the minimal norm $\minnorm{\cdot} .$
Therefore, $\M \mintensor \N$ is the \cstar-subalgebra of 
$\M \ntensor \N$ given by the norm completion of $\M \atensor \N .$

The following proposition states that 
any algebraic product map of (normal) CP channels can be 
extended to a CP channel on the maximal, minimal, binormal or normal tensor product,
which is not true for general positive maps.

\begin{prop}[%
Ref.~\onlinecite{takesakivol1}, 
Propositions~IV.4.23 and
IV.5.13
]
\label{prop:cpconti}
Let $\A , \B , \C  $ and $\D$ be \cstar-algebras.

\begin{enumerate}
\item
Let $\Phi \in \chset{\A}{\C}$ and $\Psi \in \chset{\B}{\C}$
be channels. Suppose that the ranges $\Phi (\A)$ and $\Psi (\B)$
commute.
Then the map
$
	\Phi \atensor \Psi
	\colon 
	\A \atensor \B
	\to
	\C
$
defined by
$
\Phi \atensor \Psi 
\left(
\sum_j
A_j \otimes B_j
\right)
:=
\sum_j
\Phi (A_j) \Psi (B_j)
$
$(A_j \in \A , B_j \in \B)$
uniquely extends to a channel
$\Phi \maxtensor \Psi \in \chset{\A \maxtensor \B}{\C} .$
Furthermore, if $\A ,$ $\B,$ and $\C$ are von Neumann algebras
and $\Phi$ and $\Psi$ are normal,
then $\Phi \atensor \Psi$ uniquely extends to a channel
$\Phi \btensor \Psi \in \chset{\A \btensor \B}{\C} .$
\item
Let $\Phi \in \chset{\A}{\C}$ and $\Psi \in \chset{\B}{\D}$
be channels.
Then the algebraic tensor product map
$
	\Phi \atensor \Psi
	\colon
	\A \atensor \B
	\to 
	\C \atensor \D
$
uniquely extends to a channel
$\Phi \mintensor \Psi \in  \chset{ \A \mintensor \B  }{  \C \mintensor \D   } .  $
If we further assume that
$\A , \B , \C ,$ and $\D$ are von Neumann algebras and 
that $\Phi$ and $\Psi$ are normal,
then  $\Phi \atensor \Psi$ uniquely extends to
a normal channel
$
\Phi \ntensor \Psi 
\in
\nchset{ \A \ntensor \B  }{ \C \ntensor \D  } .
$
\end{enumerate}
\end{prop}

\subsection{Channel concatenation relations and minimal sufficient channel}
\begin{defi}[Channel concatenation relations~\cite{kuramochi2017minimal,1751-8121-50-13-135302}]
\label{defi:concatenation}
\begin{enumerate}
\item
Let $\A_1 , \A_2 ,$ and $\Ain$ be \cstar-algebras
and let $\Phi_1 \in \chset{\A_1}{\Ain}$
and
$\Phi_2 \in \chset{\A_2}{\Ain}$
be channels with the common input space $\Ain .$
\begin{enumerate}[(i)]
\item
$\Phi_1$ is a \emph{concatenation} of $\Phi_2 ,$
written as $\Phi_1 \cocp \Phi_2 ,$
if there exists a channel $\Psi \in \chset{\A_1}{\A_2}$
such that $\Phi_1 = \Phi_2 \circ \Psi .$
\item
$\Phi_1$ and $\Phi_2$ are said to be \emph{concatenation equivalent},
written as
$\Phi_1 \eqcp \Phi_2 ,$
if both $\Phi_1 \cocp \Phi_2$ and $\Phi_2 \cocp \Phi_1$ hold. 
\end{enumerate}

\item
Let $\M_1,$ $\M_2$ and $\Min$ be von Neumann algebras
and let
$\Lambda_1 \in \nchset{\M_1}{\Min}$
and
$\Lambda_2 \in \nchset{\M_2}{\Min}$
be normal channels
with the common input space $\Min .$
\begin{enumerate}[(i)]
\item
$\Lambda_1$ is a \emph{normal concatenation} of $\Lambda_2 ,$
written as $\Lambda_1 \concp \Lambda_2 ,$
if there exists a normal channel $\Gamma \in \nchset{\M_1}{\M_2}$
such that $\Lambda_1 = \Lambda_2 \circ \Gamma .$
\item
$\Lambda_1$ and $\Lambda_2$ are said to be \emph{normally concatenation equivalent},
written as
$\Lambda_1 \eqncp \Lambda_2 ,$
if both $\Lambda_1 \concp \Lambda_2$ and $\Lambda_2 \concp \Lambda_1$ hold. 
\item
$\Lambda_1$ and $\Lambda_2$ are said to be \emph{normally isomorphic},
written as
$\Lambda_1 \cong \Lambda_2 ,$
if there exists a normal isomorphism 
$\pi \colon \M_1 \to \M_2$ such that
$\Lambda_1 = \Lambda_2 \circ \pi .$
\end{enumerate}
\end{enumerate}
\end{defi}
Operationally, the relation $\Phi \cocp \Psi$ for channels $\Phi$ and $\Psi$
means that the quantum operation corresponding to $\Phi$ is obtained by 
a post-processing of $\Psi .$
Note that we are discussing the channels in the Heisenberg picture.

For a given input \cstar-algebra $\Ain$
(respectively, input von Neumann algebra $\Min $),
the relations $\cocp$ and $\eqcp$ 
(respectively, $\concp$ and $\eqncp$)
are binary preorder and equivalence relations for 
(respectively, normal) channels, respectively.
We will prove in Theorem~\ref{theo:equiv} and Corollary~\ref{coro:nnonn} that,
for normal channels,
the normal concatenation relations 
$\concp$ and $\eqncp$
in fact coincide with 
$\cocp$ and $\eqcp ,$
respectively.

Let $\Lambda_0 \in \nchset{\M_0}{\Min}$ be a normal channel 
with input and outcome von Neumann algebras $\Min$ and $\M_0 ,$
respectively.
The channel $\Lambda_0$ is called minimal sufficient~\cite{kuramochi2017minimal}
if $\Lambda_0 \circ \alpha = \Lambda_0$ implies
$\alpha = \id_{\M_0}$ for any $\alpha \in \nch{\M_0} .$
For any normal channel $\Lambda \in \nchset{\M}{\Min}$
with an arbitrary outcome von Neumann algebra $\M ,$
there exists a minimal sufficient channel $\Lambda_0$
satisfying $\Lambda_0 \eqncp \Lambda$ and such $\Lambda_0$
is unique up to normal isomorphism.
If the outcome space $\M$ of $\Lambda$ is commutative,
$\M_0$ is also commutative.
Every minimal sufficient channel is faithful.

\section{Compatibility relations for channels} \label{sec:channel}

Now we introduce compatibility relations for channels as follows.

\begin{defi}[Compatibility relations for channels]
\label{defi:chcomp}
\begin{enumerate}
\item
Let $\A , \B ,$ and $\C$ be \cstar-algebras,
let $\Phi \in \chset{\A}{\C}$ and 
$\Psi \in \chset{\B}{\C}$ be channels,
and let $\gamma $ be $\min$ or $\max .$
$\Phi$ and $\Psi$ are \emph{$\gamma$-compatible},
written as
$
	\Phi \gcomp \Psi ,
$
if there exists a channel 
$\Theta_\gamma \in \chset{ \A \gtensor \B }{ \C }$
such that $\Theta_\gamma (A \otimes \1_\B) = \Phi (A)$
and
$\Theta_\gamma (\1_\A \otimes B) = \Psi (B) $
for all $A \in \A$ and for all $B \in \B .$
Such a channel $\Theta_\gamma$ is called a
$\gamma$-joint channel of $\Phi$ and $\Psi .$

\item
Let $\M, \N ,$ and $\Min$ be von Neumann algebras
and let $\Lambda \in \nchset{\M}{\Min}$ and 
$\Gamma \in \nchset{\N}{\Min}$
be normal channels.
\begin{enumerate}[(i)]
\item
$\Lambda$ and $\Gamma $ are \emph{binormally compatible},
written as $\Lambda \bcomp \Gamma ,$
if there exists a channel 
$\Theta_{\mathrm{bin}} \in \chset{\M \btensor \N}{\Min}$
such that 
$\Theta (A \otimes \1_\N) = \Lambda (A)$
and 
$\Theta (\1_\M \otimes B) = \Gamma (B)$
for all $A \in \M$ and for all $B \in \N .$
Such a channel $\Theta_{\mathrm{bin}}$
is called a $\mathrm{bin}$-joint channel of $\Lambda$ and $\Gamma .$

\item
$\Lambda$ and $\Gamma$ are \emph{normally compatible,}
written as $\Lambda \ncomp \Gamma,$ 
if there exists a normal channel
$\Theta \in \nchset{ \M \ntensor \N  }{ \Min}$
such that  
$\Theta (A \otimes \1_\N) = \Lambda (A)$
and
$\Theta (\1_\M \otimes B) = \Gamma (B)$
for all $A \in \M$ and for all $B \in \N .$
Such $\Theta$ is called a normal joint channel of $\Lambda$ and $\Gamma .$
\end{enumerate}
\end{enumerate}
\end{defi}

In Definition~\ref{defi:chcomp}, 
$\min$- and $\max$-compatibility relations coincide
if one of the outcome spaces of the channels is nuclear.
In the operational language,
two channels are compatible if they can be realized as marginals of a joint channel.

\begin{lemm}
\label{lemm:cimplication}
\begin{enumerate}
\item
Let $\A , \B , \C,
\Phi \in \chset{\A}{\C} ,$ and 
$
\Psi \in \chset{\B}{\C}
$
be the same as in Definition~\ref{defi:chcomp}.1.
Then
$
\Phi \mincomp \Psi
$
implies
$
\Phi \maxcomp \Psi .
$

\item
Let 
$\M , \N , \Min , 
\Lambda \in \nchset{\M}{\Min}, 
$
and 
$
\Gamma \in \nchset{\N}{\Min}
$
be the same as in Definition~\ref{defi:chcomp}.2.
Then the following implications hold:
\[
	\Lambda \ncomp \Gamma
	\implies
	\Lambda \mincomp \Gamma
	\implies
	\Lambda \bcomp \Gamma 
	\implies
	\Lambda \maxcomp \Gamma .
\]
\end{enumerate}
\end{lemm}
\begin{proof}
\begin{enumerate}
\item
Assume $\Phi \mincomp \Psi$ and
let $\Thmin \in \chset{\A \mintensor \B}{\C}$ be a 
$\min$-joint channel of $\Phi$ and $\Psi .$
From the universality of the maximal tensor product,
there exists a representation
$\pi \colon \A \maxtensor \B \to \A \mintensor \B $
such that 
$\pi (A \otimes B) = A \otimes B$
$(A \in \A , B \in \B ) .$
Then $\Thmax := \Thmin \circ \pi $ is a $\max$-joint channel
of $\Phi$ and $\Psi ,$
which proves
$\Phi \maxcomp \Psi .$
\item
Assume $\Lambda \ncomp \Gamma $
and let $\overline{\Theta} \in \nchset{\M \ntensor \N }{ \Min }$
be a normal joint channel of $\Lambda$ and $\Gamma .$
Since $\M \mintensor \N$ is the \cstar-subalgebra of 
$\M \ntensor \N$ generated by $\M \atensor \N ,$
$\Thmin := \overline{\Theta} \rvert_{\M \mintensor \N} ,$
the restriction of $\overline{\Theta}$ to $\M \mintensor \N ,$
is a $\min$-joint channel of $\Lambda$ and $\Gamma .$
Thus we have $\Lambda \mincomp \Gamma . $

Assume $\Lambda \mincomp \Gamma $
and let $\Thmin \in \chset{\M \mintensor \N}{\Min}$ 
be a $\min$-joint channel of $\Lambda$ and $\Gamma .$
From the minimality of the norm $\minnorm{\cdot} ,$
there exists a representation 
$\pi_{\min} \colon \M \btensor \N \to \M \mintensor \N$
such that $\pi_{\min} (A \otimes B) = A \otimes B$
$(A \in \M , B \in \N ) .$
Then 
$\Theta_{\mathrm{bin}}:= \Theta_{\min} \circ \pi_{\min}
\in \chset{\M \btensor \N}{\Min}
$
is a $\mathrm{bin}$-joint channel of $\Lambda$ and $\Gamma ,$
which proves $\Lambda \bcomp \Gamma .$

The implication 
$
	\Lambda \bcomp \Gamma 
	\implies
	\Lambda \maxcomp \Gamma
$
can be shown in the same way as in the claim 1.
\qedhere
\end{enumerate}
\end{proof}

The following lemma guarantees that the compatibility relations 
in Definition~\ref{defi:chcomp} are consistent with the 
concatenation relations for channels.

\begin{lemm}
\label{lemm:ccconsis}
\begin{enumerate}
\item
Let $\A_1 , \A_2 , \B_1 , \B_2 , $ and $\C$ be 
\cstar algebras, 
let
$\Phi_j \in \chset{\A_j}{\C}$
and
$\Psi_j \in \chset{\B_j}{\C}$
$(j=1,2)$
be channels,
and let $\gamma$ be $\min$ or $\max .$
\begin{enumerate}[(i)]
\item
Suppose that we have
$\Phi_2 \cocp \Phi_1$
and
$\Psi_2 \cocp \Psi_1 .$
Then 
$\Phi_1 \gcomp \Psi_1$ implies
$\Phi_2 \gcomp \Psi_2 .$
\item
Suppose that we have
$\Phi_2 \eqcp \Phi_1$
and
$\Psi_2 \eqcp \Psi_1 .$
Then 
$\Phi_1 \gcomp \Psi_1$ 
if and only if
$\Phi_2 \gcomp \Psi_2 .$
\end{enumerate}
\item
Let $\M_1 , \M_2 , \N_1 , \N_2 , $ and $\Min$ be 
von Neumann algebras
and let
$\Lambda_j \in \nchset{\M_j}{\Min}$
and
$\Gamma_j \in \nchset{\N_j}{\Min}$
$(j=1,2)$
be normal channels.
\begin{enumerate}[(i)]
\item
Suppose that we have
$\Lambda_2 \concp \Lambda_1$
and
$\Gamma_2 \concp \Gamma_1 .$
Then 
$\Lambda_1 \bcomp \Gamma_1$
(respectively, $\Lambda_1 \ncomp \Gamma_1$)
implies
$\Lambda_2 \bcomp \Gamma_2  $
(respectively, $\Lambda_2 \ncomp \Gamma_2  $).
\item
Suppose that we have
$\Lambda_1 \eqncp \Lambda_2$
and
$\Gamma_1 \eqncp \Gamma_2 .$
Then 
$\Lambda_1 \bcomp \Gamma_1$
(respectively, $\Lambda_1 \ncomp \Gamma_1$)
if and only if
$\Lambda_2 \bcomp \Gamma_2 $
(respectively, $\Lambda_2 \ncomp \Gamma_2 .$).
\end{enumerate}
\end{enumerate}
\end{lemm}
\begin{proof}
We first show the claim~1.(i).
Assume 
$\Phi_2 \cocp \Phi_1, $
$\Psi_2 \cocp \Psi_1 ,$
and 
$\Phi_1 \gcomp \Psi_1 .$
Then we can take channels 
$\alpha \in \chset{\A_2}{\A_1} ,$
$\beta \in \chset{\B_2}{\B_1} ,$
$\Theta_1 \in \chset{\A_1 \gtensor \B_1}{\C}$
such that
$\Phi_2 = \Phi_1 \circ \alpha ,$
$\Psi_2 = \Psi_1 \circ \beta ,$
and 
$\Theta_1$ is a $\gamma$-joint channel of $\Phi_1$ and $\Psi_1 .$
From Proposition~\ref{prop:cpconti},
$\alpha \atensor \beta$ extends to a CP channel
$\alpha \gtensor \beta \in \chset{\A_2 \gtensor \B_2}{\A_1 \gtensor \B_1} .$
Then the channel 
$\Theta_2 := \Theta_1 \circ (\alpha \gtensor \beta)$
is a $\gamma$-joint channel of $\Phi_2$ and $\Psi_2 ,$
which proves $\Phi_2 \gcomp \Psi_2 .$

The claim~1.(ii) is immediate from~1.(i).

The claim~2 can be shown similarly by using  Proposition~\ref{prop:cpconti}.
Here the post-processing channels $\alpha$ and $\beta$ are taken to be normal.
\end{proof}

\section{Universality of conjugate channel} \label{sec:main}
In this section we consider (normal) channels with a fully quantum input space
$\LHin  $ for a fixed Hilbert space $\Hin .$

\subsection{Universality of commutant conjugate channels for normal channels}

Generalizing the conjugate channel for finite-dimensional quantum-quantum 
channel~\cite{holevo2007complementary,king2007properties,1751-8121-50-13-135302},
we introduce the concept of the commutant conjugate channel 
for channels with quantum input and general outcome \cstar-algebras as follows.
\begin{defi}[Commutant conjugate channel]
\label{defi:cconj}
Let $\A$ be a \cstar-algebra
and let $\Phi \in \chset{\A}{\LHin}$ be a channel.
Take a Stinespring representation 
$(\cK, \pi , V)$ of $\Phi .$
We define the \emph{commutant conjugate channel}
of $\Phi$ associated with $(\cK, \pi , V)$
by the normal channel
$\Psi \in \nchset{\pi(\A)^\prime}{ \LHin }$
given by 
$\Psi ( B ) := V^\ast B V$
$(B \in \pi (A)^\prime) .$
\end{defi}

For a given channel, 
its commutant conjugate channel is unique up to normal concatenation equivalence 
as shown in the following proposition.

\begin{prop}
\label{prop:cconj1}
Let 
$\A,$ 
$\Hin , $
and $\Phi \in \chset{ \A }{\LHin}$
be the same as in Definition~\ref{defi:cconj}.
Then any commutant conjugate channels of $\Phi$ are
mutually normally concatenation equivalent.
\end{prop}
\begin{proof}
Let 
$(\cK_0 , \pi_0 , V_0)$
be a minimal Stinespring representation of $\Phi ,$
let
$(\cK_1 , \pi_1 , V_1)$
be an arbitrary Stinespring representation of $\Phi ,$
and let 
$\Psi_0 \in \nchset{ \pi_0 (\A)^\prime  }{ \LHin  }$
and
$\Psi_1 \in \nchset{ \pi_1 (\A)^\prime  }{ \LHin  }$
be the corresponding commutant conjugate channels, respectively.
It is sufficient to show $\Psi_0 \eqncp \Psi_1 .$
From the minimality of 
$(\cK_0 , \pi_0 , V_0) ,$
there exists an isometry $W \colon \cK_0 \to \cK_1$
such that
$WV_0 = V_1 $
and
$W \pi_0 (A) = \pi_1 (A) W$
$(A \in \A) .$
Then for any 
$A \in \A ,$ 
$B \in \pi_0 (\A)^\prime ,$
and $C \in \pi_1 (\A)^\prime , $
we have
\begin{gather*}
	\pi_1 (A) WBW^\ast 
	=
	W \pi_0 (A) BW^\ast
	=
	W B  \pi_0 (A) W^\ast
	=
	W B W^\ast \pi_1 (A) , 
	\\
	W^\ast C W \pi_0 (A)
	=
	W^\ast C  \pi_1 (A) W
	=
	W^\ast   \pi_1 (A)  C W
	=
	\pi_0 (A) W^\ast C W ,
\end{gather*}
which implies 
$WBW^\ast \in \pi_1 (\A)^\prime$
and
$W^\ast C W \in \pi_0 (\A)^\prime .$
Thus we may define normal channels 
$\Theta^W \in \nchset{ \pi_0 (\A)^\prime  }{  \pi_1 (\A)^\prime }$
and
$\Gamma^W \in \nchset{ \pi_1 (\A)^\prime  }{  \pi_0 (\A)^\prime }$
by
\begin{gather*}
	\Theta^W (B)
	:=
	W B W^\ast 
	+ 
	\phi_0 (B) 
	(\1_{\cK_1}  - WW^\ast ) ,
	\quad
	(B \in  \pi_0 (\A)^\prime) ,
	\\
	\Gamma^W 
	(C)
	:=
	W^\ast C W ,
	\quad
	(C \in \pi_1 (\A)^\prime ) ,
\end{gather*}
where $\phi_0$ is a fixed normal state on $\pi_0 (\A)^\prime .$
Then for any 
$B \in  \pi_0 (\A)^\prime $
and
$C \in \pi_1 (\A)^\prime ,$
it holds that
\begin{align*}
	\Psi_1
	\circ
	\Theta^W 
	(B)
	&=
	V_1^\ast
	WBW^\ast V_1
	+
	\phi_0 (B)
	V_1^\ast
	(\1_{\cK_1}  - WW^\ast )
	V_1
	\\
	&=
	V_0^\ast B V_0
	+ 
	\phi_0 (B)
	( V_1^\ast V_1 - V_0^\ast V_0 )
	\\
	&=
	\Psi_0 (B) ,
	\\
	\Psi_0
	\circ
	\Gamma^W 
	(C)
	&=
	V_0^\ast W^\ast C W V_0
	=
	V_1^\ast C V_1 
	=
	\Psi_1 (C) .
\end{align*}
Therefore we have $\Psi_0 \eqncp \Psi_1$
and the claim holds.
\end{proof}

In most of the following discussions,
particular choice of commutant conjugate channel of a channel
$\Psi \in \chset{\A}{\LHin} $ is irrelevant, 
and from now on we denote by $\Psi^c$ \textit{the} commutant conjugate channel
of $\Psi .$

The following three lemmas are needed for the proof of Theorem~\ref{theo:univbmax}.

\begin{lemm}
\label{lemm:cciscomp}
Let $\A$ be a \cstar-algebra,
let $\Phi \in \chset{\A}{\LHin}$ be a channel
with a minimal Stinespring representation 
$(\cK , \pi  , V),$
and let 
$\Phi^c \in \nchset{\pi (\A )^\prime}{\LHin}$
be the commutant conjugate channel of $\Phi .$
Then $\Phi \maxcomp \Phi^c .$
If we further assume that $\A$ is a von Neumann algebra 
and $\Phi$ is normal, then $\Phi \bcomp \Phi^c .$
\end{lemm}
\begin{proof}
We define a normal channel 
$\Lambda_0 \in \nchset{\pi (\A)^{\prime \prime}}{ \LHin}$
by
$\Lambda_0 (A) := V^\ast A V $
$(A \in \pi (\A)^{\prime \prime}) .$
Since we have 
$\Phi =  \Lambda_0 \circ \pi \cocp \Lambda_0, $ 
or 
$\Phi =  \Lambda_0 \circ \pi \concp \Lambda_0 $
if $\Phi$ is normal,
from Lemmas~\ref{lemm:cimplication} and \ref{lemm:ccconsis}
it is sufficient to show
$\Lambda_0 \bcomp \Phi^c .$
Since $\pi (\A)^{\prime \prime}$ and $\pi (\A)^\prime$ commute on $\LK ,$
there exists a representation
$\tilde{\pi} \colon \pi (\A)^{\prime \prime} \btensor \pi (\A)^\prime \to \LK$
such that
$\tilde{\pi} (A \otimes B) = AB$
$(A \in \pi (\A)^{\prime \prime} , B \in \pi (\A)^\prime ) .$
Define $\Theta \in \chset{  \pi (\A)^{\prime \prime} \btensor \pi (\A)^\prime }{ \LHin}$
by
$\Theta (C) := V^\ast \tilde{\pi} (C)  V$
$(C \in \pi (\A)^{\prime \prime} \btensor \pi (\A)^\prime ) .$
Then for each $A \in \pi (\A)^{\prime \prime}  $ and each
$ B \in \pi (\A)^\prime ,$ we have
\begin{gather*}
	\Theta (A \otimes \1_{\cK})
	=
	V^\ast A V
	=
	\Lambda_0 (A) ,
	\\
	\Theta (\1_{\cK} \otimes B)
	=
	V^\ast B V
	=
	\Phi^c (B) .
\end{gather*}
Therefore $\Theta$ is a $\mathrm{bin}$-joint channel of 
$\Lambda_0$ and $\Phi^c,$
which proves $\Lambda_0 \bcomp \Phi^c .$
\end{proof}

\begin{lemm}
\label{lemm:mbnormal}
Let $\M \gtensor \N $ 
be a \cstar-tensor product of von Neumann algebras
$\M $ and $\N $ 
and
let $\Theta \colon \M \gtensor \N \to \LHin$
be a positive linear map.
Suppose that the positive maps
\begin{gather*}
	\Lambda
	\colon 
	\M \ni A 
	\mapsto
	\Theta (A \otimes \1_\N)
	\in \LHin ,
	\\
	\Gamma 
	\colon
	\N \ni B 
	\mapsto
	\Theta (\1_\M \otimes B)
	\in \LHin
\end{gather*}
are normal.
Then 
$\Theta$ is binormal,
i.e.\
the maps
\begin{gather*}
	\Lambda_{B_1}
	\colon 
	\M \ni A 
	\mapsto
	\Theta (A \otimes B_1)
	\in \LHin ,
	\\
	\Gamma_{A_1}
	\colon
	\N \ni B 
	\mapsto
	\Theta (A_1 \otimes B)
	\in \LHin
\end{gather*}
are ultraweakly continuous for each $A_1 \in \M$
and each $B_1 \in \N .$
\end{lemm}
\begin{proof}
We show the ultraweak continuity of $\Lambda_{B_1} $
for $B_1 \in \N .$
Since $B_1$ is a linear combination of positive elements of $\N ,$
it is sufficient to show the normality of $\Lambda_{B_1}$ when $B_1 \geq 0 .$
We take an arbitrary monotonically decreasing net $A_\alpha \downarrow 0$ in $\M .$
Then we have
\[
	0 \leq
	\Lambda_{B_1} (A_\alpha)
	=
	\Theta (A_\alpha \otimes B_1)
	\leq
	\Theta (A_\alpha \otimes \norm{B_1} \1_{\N})
	=
	\norm{B_1}
	\Lambda (A_\alpha)
	\downarrow 0 ,
\]
which implies $\Lambda_{B_1} (A_\alpha ) \downarrow 0.$
Therefore $\Lambda_{B_1}$ is normal.
The ultraweak continuity of $\Gamma_{A_1} $ can be shown
in the same way.
\end{proof}

\begin{lemm}
\label{lemm:normalrep}
Let 
$\M , \N , \Min , 
\Lambda \in \nchset{\M}{\Min}, 
$
and 
$
\Gamma \in \nchset{\N}{\Min}
$
be the same as in Definition~\ref{defi:chcomp}.2
and let $\M \gtensor \N$ be a \cstar-tensor product. 
Suppose that 
$\Min$ acts on a Hilbert space $\Hin$ and that
$\Theta \in \chset{\M \gtensor \N }{\Min}
\subseteq 
\chset{\M \gtensor \N }{\LHin}
$
is a joint channel of $\Lambda$ and $\Gamma$
with a minimal Stinespring representation
$(\cK_{\Theta} , \pi_{\Theta} , V_{\Theta}) .$
Then the representations 
\begin{gather}
	\pi_\Lambda 
	\colon
	\M \ni A \mapsto 
	\pi_\Theta (A \otimes \1_\N)
	\in \cL  (\cK_{\Theta}) ,
	\label{eq:pil}
	\\
	\pi_\Gamma 
	\colon
	\N \ni B \mapsto 
	\pi_\Theta (\1_\M \otimes B)
	\in \cL  (\cK_{\Theta})
	\label{eq:pig} 
\end{gather}
are normal.
\end{lemm}
\begin{proof}
We only prove the normality of $\pi_\Lambda;$
that of $\pi_\Gamma$ can be shown similarly.
For this, it is sufficient to show that 
for an arbitrary monotonically increasing bounded net
$ 0 \leq A_\alpha \uparrow A$ in $\M ,$
$\pi_{\Lambda} (A_\alpha)$ converges to $\pi_\Lambda (A)$
in the weak operator topology.
Since $\Theta$ is binormal from Lemma~\ref{lemm:mbnormal},
for each $\phi_1 , \phi_2 \in \Hin ,$ 
each $ A_1 , A_2 \in \M ,$
and each $B_1 , B_2 \in \N , $ we have
\begin{align*}
	&\braket{  
	\pi_\Theta (A_1 \otimes B_1)  
	V_{\Theta}  
	\phi_1
	|
	\pi_\Lambda (A_\alpha)
	\pi_\Theta (A_2 \otimes B_2)  
	V_{\Theta}  
	\phi_2
	}
	\\
	&=
	\braket{    
	\phi_1
	|
	V_\Theta^\ast
	\pi_\Theta (A_1^\ast A_\alpha A_2 \otimes B_1^\ast B_2)    
	V_\Theta
	\phi_2
	}
	\\
	&=
	\braket{    
	\phi_1
	|
	\Theta (A_1^\ast A_\alpha A_2 \otimes B_1^\ast B_2)    
	\phi_2
	}
	\\
	&\to
	\braket{    
	\phi_1
	|
	\Theta (A_1^\ast A A_2 \otimes B_1^\ast B_2)    
	\phi_2
	}
	\\
	&=
	\braket{  
	\pi_\Theta (A_1 \otimes B_1)  
	V_{\Theta}  
	\phi_1
	|
	\pi_\Lambda (A)
	\pi_\Theta (A_2 \otimes B_2)  
	V_{\Theta}  
	\phi_2
	} .
\end{align*}
Since $\pi_\Theta (\M \atensor \N) V_\Theta \Hin$ generates $\cK_\Theta$
and $(\pi_\Lambda (A_\alpha))$ is a bounded net,
this shows that $\pi_\Lambda (A_\alpha)$ 
converges to $\pi_\Lambda (A)$ 
in the weak operator topology.
\end{proof}

The following Theorems~\ref{theo:univbmax} and \ref{theo:equiv} are 
the first main results of this paper,
generalizing the result of Ref.~\onlinecite{1751-8121-50-13-135302}.

\begin{theo}
\label{theo:univbmax}
Let $\M$ and $\N$ be von Neumann algebras
and let $\Lambda \in \nchset{\M}{\LHin}$
and
$\Gamma \in \nchset{\N}{\LHin}$ be
normal channels.
Then the following conditions are equivalent.
\begin{enumerate}[(i)]
\item \label{enum:1-1}
$\Lambda \bcomp \Gamma ;$
\item \label{enum:1-2}
$\Lambda \maxcomp \Gamma ;$
\item \label{enum:1-3}
$\Gamma \concp \Lambda^c ;$
\item \label{enum:1-4}
$\Gamma \cocp \Lambda^c ;$
\item \label{enum:1-5}
$\Lambda \concp \Gamma^c ;$
\item \label{enum:1-6}
$\Lambda \cocp \Gamma^c .$
\end{enumerate}
\end{theo}
\begin{proof}
(\ref{enum:1-1})$\implies$(\ref{enum:1-2}) follows from Lemma~\ref{lemm:cimplication}.

(\ref{enum:1-2})$\implies$(\ref{enum:1-3}).
Assume~(\ref{enum:1-2}) and let
$\Theta \in \chset{\M \maxtensor \N}{\LHin}$
be a $\max$-joint channel of $\Lambda$ and $\Gamma$ 
with a minimal Stinespring representation 
$(\cK_\Theta , \pi_\Theta , V_\Theta ) .$
We define representations 
$\pi_\Lambda$
and 
$\pi_\Gamma$ 
by \eqref{eq:pil} and \eqref{eq:pig},
which are normal by Lemma~\ref{lemm:normalrep}.
From 
$\Lambda (A) 
= \Theta (A \otimes \1_\N) 
=V_\Theta^\ast \pi_\Lambda (A ) V_\Theta $
$(A \in \M) ,$
$(\cK_\Theta , \pi_\Lambda , V_\Theta)$ is a Stinespring representation of $\Lambda .$
Therefore from Proposition~\ref{prop:cconj1},
we can define $\Lambda^c \in \nchset{ \pi_\Lambda (\M)^\prime  }{\LHin}$
by
$\Lambda^c (C) := V_\Theta^\ast C V_\Theta$
$(C \in \pi_\Lambda (\M)^\prime ) .$
Since $\pi_\Gamma (\N) \subseteq \pi_\Lambda (\M)^\prime$
and $\pi_\Gamma$ is normal,
we have 
$
\Gamma = \Lambda^c \circ \pi_\Gamma 
\concp
\Lambda^c ,
$
proving~(\ref{enum:1-3}).

(\ref{enum:1-3})$\implies$(\ref{enum:1-4}) is obvious.

(\ref{enum:1-4})$\implies$(\ref{enum:1-2}) 
and
(\ref{enum:1-3})$\implies$(\ref{enum:1-1})
follow from Lemmas~\ref{lemm:ccconsis} and \ref{lemm:cciscomp}.

The equivalence 
(\ref{enum:1-1})$\iff$(\ref{enum:1-2})$\iff$(\ref{enum:1-5})$\iff$(\ref{enum:1-6})
can be shown in the same way.
\end{proof}

\begin{theo}
\label{theo:equiv}
Let $\M$ and $\N$ be von Neumann algebras
and let $\Lambda \in \nchset{\M}{\LHin}$
and
$\Gamma \in \nchset{\N}{\LHin }$ be normal channels.
Then the following conditions are equivalent.
\begin{enumerate}[(i)]
\item
$\Lambda \concp \Gamma ;$
\item
$\Lambda \cocp \Gamma ;$
\item
$\Gamma^c \concp \Lambda^c ;$
\item
$\Gamma^c \cocp \Lambda^c .$
\end{enumerate}
\end{theo}
\begin{proof}
(i)$\implies$(ii) and (iii)$\implies$(iv) are obvious.

(ii)$\implies$(iii).
Assume (ii).
From $\Gamma \maxcomp \Gamma^c$
and Lemma~\ref{lemm:ccconsis},
we have $\Lambda \maxcomp \Gamma^c .$
Therefore Theorem~\ref{theo:univbmax} implies
$\Gamma^c \concp \Lambda^c .$

(iv)$\implies$(i).
Assume (iv) and
let $(\cK_\Lambda , \pi_\Lambda , V_\Lambda )$
and 
$(\cK_\Gamma , \pi_\Gamma  , V_\Gamma  )$
be minimal Stinespring representations of 
$\Lambda $ and $\Gamma ,$ respectively.
Then the conjugate channels $\Lambda^c $ and $\Gamma^c$
are the maps given by
\begin{gather*}
	\Lambda^c (C)
	=
	V_\Lambda^\ast C V_\Lambda 
	\quad
	(C \in \pi_\Lambda (\M)^\prime) ,
	\\
	\Gamma^c (D)
	=
	V_\Gamma^\ast D V_\Gamma 
	\quad
	(D \in \pi_\Gamma (\N)^\prime) .
\end{gather*}
Define 
$\Lambda^{cc} \in \nchset{\pi( \M ) }{\LHin}$
and
$\Gamma^{cc} \in \nchset{\pi( \N ) }{\LHin}$
by
\begin{gather*}
	\Lambda^{cc} (A)
	=
	V_\Lambda^\ast A V_\Lambda 
	\quad
	(A \in \pi_\Lambda (\M)) ,
	\\
	\Gamma^{cc} (B)
	=
	V_\Gamma^\ast B V_\Gamma 
	\quad
	(B \in \pi_\Gamma (\N)) .
\end{gather*}
Then $\Lambda^{cc}$ and $\Gamma^{cc}$ are commutant conjugate channels of
$\Lambda^c $ and $\Gamma^c ,$
respectively.
Hence, from the implication (ii)$\implies$(iii),
we obtain $\Lambda^{cc} \concp \Gamma^{cc} .$
Thus it is sufficient to prove $\Lambda \eqncp \Lambda^{cc}$
and $\Gamma \eqncp \Gamma^{cc} .$

Now we show 
$\Lambda \eqncp \Lambda^{cc} .$
By following the construction given in Ref.~\onlinecite{kuramochi2017minimal}
(Lemma~2),
there exist a von Neumann algebra $\M_0$
and a normal channel 
$\Lambda_0 \in \nchset{\M_0}{\LHin}$
such that
$\Lambda \eqncp \Lambda_0$
and
$\Lambda_0$ is faithful.
We take a minimal Stinespring representation 
$(\cK_0 , \pi_0 , V_0)$ of $\Lambda_0$
and define conjugate channels 
$\Lambda_0^c \in \nchset{\pi_0 (\M_0)^\prime}{\LHin}$
and
$\Lambda_0^{cc} \in \nchset{\pi_0 (\M_0)}{\LHin}$
in the same way as $\Lambda^c$ and $\Lambda^{cc} .$
Then from the implication (ii)$\implies$(iii)
and $\Lambda \eqncp \Lambda_0 ,$
we have
$\Lambda^c \eqncp \Lambda_0^c$
and
$\Lambda^{cc} \eqncp \Lambda_0^{cc} .$
On the other hand, 
since $\pi_0$ is a normal isomorphism 
from $\M_0$ onto $\pi_0 (\M_0)$
due to the faithfulness of $\Lambda_0 ,$
$\Lambda_0$ and $\Lambda_0^{cc}$ are isomorphic, which implies
$\Lambda_0 \eqncp \Lambda_0^{cc} .$
Therefore we obtain $\Lambda \eqncp \Lambda_0 \eqncp \Lambda_0^{cc} \eqncp \Lambda^{cc} .$
We can show
$\Gamma \eqncp \Gamma^{cc} $
in the same way,
which completes the proof of (iv)$\implies$(i).
\end{proof}

From the above proof, we obtain
\begin{coro}
\label{coro:cc}
Let $\M$ be a von Neumann algebra, 
let $\Lambda \in \nchset{\M}{\LHin}$
be a normal channel,
and let $\Lambda^{cc}$ be a conjugate channel
of $\Lambda^c .$
Then $\Lambda \eqncp \Lambda^{cc} .$
\end{coro}

\begin{rem}
\label{rem:class}
For a normal channel 
$\Lambda \in \nchset{\M}{\LHin},$
we denote by
$\mathfrak{C}_\Lambda^{\max}$
the class of normal channels $\max$-compatible with $\Lambda .$
Then Theorem~\ref{theo:univbmax} states that
$\Gamma \in \mathfrak{C}_\Lambda^{\max}$ 
if and only if 
$\Gamma \concp \Lambda^c $
for a normal channel $\Gamma .$
In this sense, $\Lambda^c$ is the maximal element of 
$\mathfrak{C}^{\max}_\Lambda$
with respect to the concatenation preorder.
Theorem~\ref{theo:equiv}, on the other hand,
implies that
$\Lambda \concp \Gamma$
if and only if
$
\mathfrak{C}^{\max}_\Gamma 
\subseteq 
\mathfrak{C}^{\max}_\Lambda 
$
for normal channels $\Lambda$ and $\Gamma .$
Thus the equivalence class of a normal channel $\Lambda$
with respect to the concatenation equivalence relation $\eqncp$ or $\eqcp$
is uniquely characterized by $\mathfrak{C}^{\max}_\Lambda .$
\end{rem}

The equivalence of the relations
$\cocp$ and $\concp$
(respectively, $\eqcp$ and $\eqncp$)
for normal channels in Theorem~\ref{theo:equiv}
can be established for general input von Neumann algebras 
as in the following corollary.
\begin{coro}
\label{coro:nnonn}
Let $\M ,$ $\N , $ and $\Min$
be von Neumann algebras
and 
let 
$\Lambda \in \nchset{\M}{\Min}$
and
$\Gamma \in \nchset{\N}{\Min}$
be normal channels.
Then 
$\Lambda \cocp \Gamma$
(respectively, $\Lambda \eqcp \Gamma$)
if and only if
$\Lambda \concp \Gamma$
(respectively, $\Lambda \eqncp \Gamma$).
\end{coro}
\begin{proof}
Suppose that $\Min$ acts on a Hilbert space $\Hin .$
Then we may regard $\Lambda$ and $\Gamma$ as normal channels
in $\nchset{\M}{\LHin}$
and 
$\nchset{\N}{\LHin} ,$
respectively.
Thus the claim is immediate from Theorem~\ref{theo:equiv}.
\end{proof}
We can further generalize this equivalence to statistical experiments
with an arbitrary parameter set as follows.

A triple $\E = \seE$ is called 
a (quantum) statistical 
experiment~\cite{jencovapetz2006,gutajencova2007,kuramochi2017minimal}
if $\M$ is a von Neumann algebra called the outcome space of $\E$,
$\Theta$ is a set called the parameter set of $\E$,
and $(\phth)_{\thin} \in \Ss (\M)^\Theta$
is a family of normal states on $\M$ parametrized by $\Theta .$
For a pair of statistical experiments
$\E = \seE$ and $\F = \seF $ with a common parameter set $\Theta ,$
we define the following coarse-graining preorder and equivalence relations:
\begin{align*}
	\E \cocp \F
	&:\defarrow
	\exists \Gamma 
	\in
	\chset{\M}{\N}
	\text{ s.t.\ }
	\left[
	\phth = \psth \circ \Gamma 
	\, (\forall \thin)
	\right] ;
	\\
	\E \concp \F
	&:\defarrow
	\exists \Gamma 
	\in
	\nchset{\M}{\N}
	\text{ s.t.\ }
	\left[
	\phth = \psth \circ \Gamma 
	\, (\forall \thin)
	\right] ;
	\\
	\E \eqcp \F
	&:\defarrow
	\text{$\E \cocp \F$ and $\F \cocp \E ;$}
	\\
	\E \eqncp \F
	&:\defarrow
	\text{$\E \concp \F$ and $\F \concp \E .$}
\end{align*}
For each statistical experiment
$\E = \seE ,$
we define the associated normal channel 
$\Lambda_{\E} \in \nchset{\M}{\cL  (\ell^2 (\Theta))} $
by
\[
	\Lambda_\E (A)
	:=
	\sum_{\thin }
	\phth (A)
	\ket{\delta_\theta}
	\bra{\delta_\theta} ,
\]
where $\ell^2 (\Theta)$ is the Hilbert space of square-summable 
complex-valued functions on $\Theta$
equipped with the inner product
\[
	\braket{f|g}
	:=
	\sum_{\thin}
	\overline{f(\theta)} g(\theta )
	\quad
	(f,g \in \ell^2 (\Theta))
\]
and $(\delta_\theta)_{\thin}$ is an orthonormal basis given by
\[
	\delta_\theta (\theta^\prime)
	:=
	\begin{cases}
		1 , & (\theta^\prime = \theta) ; \\
		0 , & (\theta^\prime \neq \theta) .
	\end{cases}
\]
\begin{coro}
\label{coro:se}
Let $\E = \seE$ and $\F = \seF$ be statistical experiments
with a common parameter set $\Theta .$
Then the following conditions are equivalent.
\begin{enumerate}[(i)]
\item
$\E \cocp \F$
(respectively, $\E \eqcp \F$);
\item
$\E \concp \F$
(respectively, $\E \eqncp \F$);
\item
$\Lambda_\E \cocp \Lambda_\F$
(respectively, $\Lambda_\E \eqcp \Lambda_\F$);
\item
$\Lambda_\E \concp \Lambda_\F$
(respectively, $\Lambda_\E \eqncp \Lambda_\F$).
\end{enumerate}
\end{coro}
\begin{proof}
(i)$\iff$(iii) and (ii)$\iff$(iv). For a linear map 
$\Gamma \colon \M \to \N ,$
we have the following equivalences:
\begin{align*}
	\Lambda_\E = \Lambda_\F \circ \Gamma
	&\iff
	\sum_{\thin}
	\phth (A) 
	\ket{\delta_\theta} \bra{\delta_\theta}
	=
	\sum_{\thin}
	\psth \circ \Gamma (A) 
	\ket{\delta_\theta} \bra{\delta_\theta} ,
	\quad (\forall A \in \M)
	\\
	&\iff
	\phth (A) = \psth \circ \Gamma (A) ,
	\quad (\forall \thin , \forall A \in \M)
	\\
	&\iff
	\phth = \psth \circ \Gamma
	, \quad
	(\forall \thin) .
\end{align*}
Therefore we obtain the equivalences
(i)$\iff$(iii) and (ii)$\iff$(iv).

(iii)$\iff$(iv) is immediate from Theorem~\ref{theo:equiv}.
\end{proof}

\subsection{%
Universality of the commutant conjugate channels for channels with outcome \cstar-algebras%
}
Now we consider the $\max$-compatibility
relation for channels with arbitrary outcome \cstar-algebras.

We first define the normal extension of a channel
whose existence and uniqueness are assured in the following lemma.
\begin{lemm}
\label{lemm:normalextension}
Let $\A$ be a \cstar-algebra,
let $\Min$ be a von Neumann algebra on $\Hin ,$
and let $\Phi \in \chset{\A}{\Min}$ be a channel.
Then $\Phi$ uniquely extends to a normal channel
$\overline{\Phi} \in \nchset{\A^{\ast \ast} }{\Min} .$
The channel $\overline{\Phi}$ is called the 
normal extension of $\Phi .$
\end{lemm}
\begin{proof}
Let $(\cK , \pi , V)$ be a minimal Stinespring representation of 
$\Phi .$
From the universality of the enveloping von Neumann algebra
$\A^{\ast \ast} ,$
$\pi$ extends to a normal representation
$\overline{\pi} \colon \A^{\ast \ast} \to \LK .$
Then we define $\overline{\Phi} \colon \A^{\ast \ast } \to \LHin$
by
$\overline{\Phi} (A) = V^\ast \overline{\pi} (A) V$
$(A \in \A^{\ast \ast}),$
which is a normal extension of $\Phi .$
Since $\A$ is ultraweakly dense in $\A^{\ast \ast}$
and $\overline{\Phi} (\A) = \Phi (\A) \subseteq \Min ,$
$\overline{\Phi}$ is a unique channel in $\nchset{\A^{\ast \ast}}{\Min}$
that is an extension of $\Phi .$
\end{proof}

The following lemma states that
the normal extension $\overline{\Phi}$ of a channel $\Phi$
is the smallest normal channel greater than $\Phi$
with respect to the concatenation preorder.

\begin{lemm}
\label{lemm:minex}
Let $\A , \Min , \Phi \in \chset{\A}{\Min} ,$
and $\overline{\Phi} \in \nchset{\A^{\ast \ast }}{ \Min}$
be the same as in Lemma~\ref{lemm:normalextension},
let $\M$ be a von Neumann algebra,
and let $\Lambda \in \nchset{\M}{\Min}$ be a normal channel.
Then $\Phi \cocp \Lambda$ if and only if 
$\overline{\Phi} \concp \Lambda .$
\end{lemm}
\begin{proof}
Assume $\Phi \cocp \Lambda $
and let $\alpha \in \chset{\A}{\M}$ be a channel satisfying
$\Phi = \Lambda \circ \alpha .$
Then Lemma~\ref{lemm:normalextension} implies that
$\alpha$ uniquely extends to
a normal channel 
$\overline{\alpha} \in \nchset{\A^{\ast \ast} }{\M} .$
Thus
$
\overline{\Phi}_1 := \Lambda \circ \overline{\alpha}
\in \nchset{\A^{\ast \ast}}{\Min}
$
is a normal channel such that 
$\overline{\Phi}_1 (A) = \Lambda \circ \alpha (A) = \Phi(A) = \overline{\Phi}(A)$
for all $A \in \A .$
Since $\A$ is ultraweakly dense in $\A^{\ast \ast} ,$
this implies
$\overline{\Phi} =\overline{\Phi}_1 
= \Lambda \circ \overline{\alpha}
\concp \Lambda .
$
The converse implication is immediate from 
$\Phi \cocp \overline{\Phi} .$
\end{proof}

\begin{coro}
\label{coro:nexn}
Let $\M$ be a von Neumann algebra, 
let $\Lambda \in \nchset{\M}{\LHin}$
be a normal channel,
and let $\overline{\Lambda} \in \nchset{\M^{\ast \ast}}{\LHin}$ be the
normal extension of $\Lambda .$
Then we have $\Lambda \eqncp \overline{\Lambda}  .$
\end{coro}
\begin{proof}
From $\Lambda \cocp \Lambda $ and Lemma~\ref{lemm:minex}, 
we have $\overline{\Lambda} \concp \Lambda .$
On the other hand, we have $\Lambda \cocp \overline{\Lambda} $ by definition,
and hence Theorem~\ref{theo:equiv} implies $\Lambda \concp \overline{\Lambda} .$
Thus we have $\Lambda \eqncp \overline{\Lambda}  .$
\end{proof}

The universality of the commutant conjugate channel is 
still valid for channels with general outcome \cstar-algebras
as in the following theorem.

\begin{theo}
\label{theo:univmax}
Let $\A $ and $\B$ be \cstar-algebras,
let $\Phi \in \chset{\A}{\LHin}$ and 
$\Psi \in \chset{\B}{\LHin}$ be channels,
and let 
$\overline{\Phi} \in \nchset{\A^{\ast \ast}}{\LHin}$
and
$\overline{\Psi} \in \nchset{\B^{\ast \ast}}{\LHin}$
be the normal extensions of $\Phi$ and $\Psi ,$
respectively.
Then the following conditions are equivalent.
\begin{enumerate}[(i)]
\item \label{enum:3-1}
$\Phi \maxcomp \Psi ;$
\item \label{enum:3-2}
$\overline{\Phi} \maxcomp \overline{\Psi} ;$
\item \label{enum:3-3}
$\overline{\Phi} \bcomp \overline{\Psi} ;$
\item \label{enum:3-4}
$\Psi \cocp \Phi^c ;$
\item \label{enum:3-5}
$\overline{\Psi} \concp \Phi^c ;$
\item \label{enum:3-6}
$\Phi \cocp \Psi^c ;$
\item \label{enum:3-7}
$\overline{\Phi} \concp \Psi^c .$
\end{enumerate}
\end{theo}
\begin{proof}
(\ref{enum:3-1})$\implies$(\ref{enum:3-2}).
Assume~(\ref{enum:3-1}) and let 
$\Theta \in \chset{\A \maxtensor \B}{\LHin}$
be a $\max$-joint channel of $\Phi$ and $\Gamma$
with a minimal Stinespring representation
$(\cK_\Theta , \pi_\Theta , V_\Theta ) .$
Define representations
\begin{gather*}
	\pi_\Phi
	\colon
	\A \ni A
	\mapsto 
	\pi_\Theta 
	(A \otimes \1_\B)
	\in
	\cL (\cK_\Theta),
	\\
	\pi_\Psi
	\colon
	\B \ni B
	\mapsto 
	\pi_\Theta 
	(\1_\A \otimes B)
	\in
	\cL (\cK_\Theta),
\end{gather*}
whose ranges commute on $\cL (\cK_\Theta) .$
Then from the universality of the enveloping von Neumann algebras
$\A^{\ast \ast}$ and $\B^{\ast \ast} ,$
$\pi_\Phi$ and $\pi_\Psi$
extend to normal representations
$\overline{\pi}_{\Phi} \colon \A^{\ast \ast} \to \cL (\cK_\Theta)$
and
$\overline{\pi}_{\Psi} \colon \B^{\ast \ast} \to \cL (\cK_\Theta) ,$
respectively.
From 
$\overline{\Phi} (A) 
= V_\Theta^\ast 
\overline{\pi}_{\Phi}
(A)
V_\Theta
$
($A\in \A$)
and
$\overline{\Psi} (B) 
= V_\Theta^\ast 
\overline{\pi}_{\Psi}
(B)
V_\Theta
$
($B\in \B$),
$(\cK_\Theta , \overline{\pi}_\Phi , V_\Theta)$
and
$(\cK_\Theta , \overline{\pi}_\Psi , V_\Theta)$
are Stinespring representations of $\overline{\Phi}$
and $\overline{\Psi} ,$
respectively.
Since the ranges 
$\overline{\pi}_{\Phi} (\A^{\ast \ast}) = \pi_{\Phi} (\A)^{\prime \prime}$
and
$\overline{\pi}_{\Psi} (\B^{\ast \ast} ) = \pi_{\Psi} (\B)^{\prime \prime}$
also commute,
there exists a representation 
$
\tilde{\pi} \colon 
\A^{\ast \ast } \maxtensor \B^{\ast \ast}
\to
\cL (\cK_\Theta)
$
such that
$\tilde{\pi} (A^{\prime \prime} \otimes B^{\prime \prime} ) = 
\overline{\pi}_{\Phi} (A^{\prime \prime})
\overline{\pi}_{\Psi} (B^{\prime \prime})
$
($A^{\prime \prime}\in \A^{\ast \ast} , 
B^{\prime \prime} \in \B^{\ast \ast}$).
Now we define a channel
$\widetilde{\Theta} \in 
\chset{\A^{\ast \ast } \maxtensor \B^{\ast \ast }}{ \LHin}$
by
$\widetilde{\Theta} (X) = V_\Theta^\ast \tilde{\pi} (X) V_\Theta $
$(X \in \A^{\ast \ast} \maxtensor \B^{\ast \ast}) .$
Then 
for each 
$A^{\prime \prime}\in \A^{\ast \ast} $
and each 
$B^{\prime \prime} \in \B^{\ast \ast}$
we have
\begin{gather*}
	\widetilde{\Theta}
	(A^{\prime \prime} \otimes \1_{\B^{\ast \ast}})
	=
	V_{\Theta}^\ast
	\overline{\pi}_{\Phi} (A^{\prime \prime})
	V_\Theta
	=
	\overline{\Phi} (A^{\prime \prime}) ,
	\\
	\widetilde{\Theta}
	(\1_{\A^{\ast \ast} }\otimes B^{\prime \prime})
	=
	V_{\Theta}^\ast
	\overline{\pi}_{\Psi} (B^{\prime \prime})
	V_\Theta
	=
	\overline{\Psi} (B^{\prime \prime}) .
\end{gather*}
Thus $\widetilde{\Theta}$ is a $\max$-joint channel
of $\overline{\Phi}$ and $\overline{\Psi} ,$
which proves 
$\overline{\Phi} 
\maxcomp
\overline{\Psi} .$

(\ref{enum:3-2})$\implies$(\ref{enum:3-1})
follows from $\Phi \cocp \overline{\Phi} ,$
$\Psi \cocp \overline{\Psi},$
and Lemma~\ref{lemm:ccconsis}.

The equivalence
(\ref{enum:3-2})$\iff$(\ref{enum:3-3})$\iff$(\ref{enum:3-5})$\iff$(\ref{enum:3-7})
follows from Theorem~\ref{theo:univbmax}
by noting that the commutant conjugate channels 
$\overline{\Phi}^c$ and $\overline{\Psi}^c$
coincide with 
$\Phi^c$ and $\Psi^c ,$
respectively.

The equivalences 
(\ref{enum:3-4})$\iff$(\ref{enum:3-5})
and
(\ref{enum:3-6})$\iff$(\ref{enum:3-7})
are immediate from Lemma~\ref{lemm:minex}.
\end{proof}

\section{Compatibility of POVMs} \label{sec:povm}
In this section, we consider compatibility of POVM and channel
by identifying each POVM with the corresponding QC channel.
Throughout this section, we again consider the fully quantum input space
$\LHin .$

\subsection{POVM and QC channel}
A POVM on $\Hin$ is a triple $(\Omega , \Sigma , \oM)$
such that $(\Omega , \Sigma )$ is a measurable space
and
$\oM \colon \Sigma \to \LHin$ is a mapping satisfying
\begin{enumerate}[(i)]
\item
$\oM(E) \geq 0 $
for all $E \in \Sigma ;$
\item
$\oM(\Omega) = \1_{\Hin} ;$
\item
for any disjoint sequence $\{ E_n \} \subseteq \Sigma ,$
$\oM (\cup_n E_n) = \sum_n \oM (E_n)$
in the weak operator topology.
\end{enumerate}
If $\oM (E) $ is a projection for each $E \in \Sigma ,$
$\oM$ is called a projection valued measure (PVM).

For each density operator $\rho \in \SHin ,$
we define the outcome probability measure $P^{\oM}_\rho$
on $(\Omega , \Sigma)$ by
$P^{\oM}_\rho (E) := \tr [ \rho \oM (E)  ]$
$(E \in \Sigma) .$
Operationally, a POVM describes the statistics of the classical outcome 
of a general quantum measurement.

Let $(\Omega , \Sigma , \oM)$ be a POVM on $\Hin .$
A triple $(\cK , \oP , V)$ is called a
Naimark dilation~\cite{Naimark1940,davies1976quantum} 
of $\oM$ if
$\cK$ is a Hilbert space,
$(\Omega , \Sigma , \oP)$ is a PVM on $\cK ,$
and $V \colon \Hin \to \cK$ is a linear isometry
such that
$\oM(E) = V^\ast \oP(E) V$
for all $E \in \Sigma .$
A Naimark dilation $(\cK ,\oP , V)$ is called minimal
if the closed linear span of $\oP(\Sigma) V \Hin  $
coincides with $\cK .$
Naimark's dilation theorem~\cite{Naimark1940,davies1976quantum} 
states that 
every POVM has a minimal Naimark dilation unique up to unitary equivalence.

For a measurable space $(\Omega , \Sigma) ,$
we denote by $B(\Omega , \Sigma)$
the set of $\Sigma$-measurable, bounded, complex-valued functions on $\Omega .$
$B(\Omega , \Sigma)$ is a commutative \cstar-algebra
under the supremum norm
$\norm{f} := \sup_{\omega \in \Omega} |f(\omega)| .$

\begin{defi}[QC channel]
\label{defi:qcch}
Let $(\Omega , \Sigma , \oM)$ be a POVM on $\Hin  ,$
let $\gamma^\oM \in \chset{B(\Omega , \Sigma)}{\LHin}$
be the channel defined by
\[
	\gamma^\oM
	(f)
	:=
	\int_\Omega 
	f(\omega)
	d\oM (\omega)
	\quad
	(f \in B(\Omega , \Sigma)),
\]
called the pre-QC channel of $\oM,$
and let $(\cK , \pi , V)$ be a minimal Stinespring representation of $\gamma^\oM .$
We define the \emph{QC channel}
of $\oM$ by the normal channel
$\Gamma^\oM \in \nchset{\pi (B(\Omega , \Sigma))^{\prime \prime}}{\LHin}$
given by
$\Gamma^\oM (A) := V^\ast A V $
$(A \in \pi (B(\Omega , \Sigma))^{\prime \prime} ) .$
\end{defi}

\begin{rem}
\label{rem:qcch}
From Corollary~\ref{coro:cc}, the QC channel $\Gamma^\oM$
in Definition~\ref{defi:qcch} is normally concatenation equivalent 
to the normal extension 
$\overline{\gamma^\oM} \in \nchset{B(\Omega , \Sigma)^{\ast \ast}}{\LHin} $
of $\gamma^\oM .$
Moreover, if we define $\oP (E) := \pi (\chi_E)$
$(E \in \Sigma) ,$
where $\chi_E$ is the indicator function of $E ,$
$(\cK , \oP , V)$ is a minimal Naimark dilation of $\oM .$
(Indeed, by this construction of $\oP$, Naimark's dilation theorem 
immediately follows from Stinespring's dilation theorem~\cite{1955stinespring}).
\end{rem}

\begin{lemm}
\label{lemm:qcconj}
Let $(\Omega , \Sigma , \oM)$ be a POVM on $\Hin .$
Then $\Gamma^{\oM} \concp (\Gamma^{\oM})^c . $
\end{lemm}
\begin{proof}
Let $\A_{\oM}$ be the outcome space of $\Gamma^\oM .$
Since $\A_\oM $ is a commutative von Neumann algebra,
there exists a representation 
$\tilde{\pi} \colon \A_\oM \maxtensor \A_\oM \to \A_\oM$
such that
$\tilde{\pi} (A_1 \otimes A_2) = A_1 A_2 $
$(A_1 , A_2 \in \A_\oM) .$
Then $\Theta \in \chset{\A_\oM \maxtensor \A_\oM}{\LHin}$
defined by
$\Theta := \Gamma^{\oM} \circ \tilde{\pi}$
is a $\max$-joint channel of $\Gamma^{\oM} $
and $\Gamma^{\oM} .$
Therefore $\Gamma^{\oM} \concp (\Gamma^\oM)^c$ follows
from Theorem~\ref{theo:univbmax}.
\end{proof}

For a localizable~\cite{segal1951} measure space
$(\Omega, \Sigma , \mu) ,$
we denote the $L^\infty$ space of
$(\Omega, \Sigma , \mu) $ by $L^\infty (\mu)  .$
The notion of $\mu$-almost everywhere ($\mu$-a.e.) equality 
defines an equivalence relation for $\Sigma$-measurable 
functions and the equivalence class to which a measurable function $f$ belongs 
is denoted by $[f]_\mu .$
$L^\infty (\mu)$ is a commutative von Neumann algebra acting on the Hilbert space
$L^2 (\mu).$

The following proposition assures that 
the definition of the QC channel in Ref.~\onlinecite{kuramochi2017minimal}
for separable input Hilbert space $\Hin$ coincides with the present one up to normal isomorphism.

\begin{prop}
\label{prop:qceq}
Let $(\Omega , \Sigma , \oM)$ be a POVM 
on $\Hin .$
Suppose that $\Hin$ is separable and take a faithful density operator
$\rho_0 \in \SHin .$
Define a normal channel 
$\Gamma^\oM_0 \in \nchset{ L^\infty (\mu) }{\LHin}$
by
\begin{gather}
	\mu :=
	P^{\oM}_{\rho_0} ,
	\notag
	\\
	\Gamma_0^\oM ([f]_\mu)
	:=
	\int_{\Omega} f (\omega)
	d\oM (\omega) .
	\notag
\end{gather}
Then $\Gamma^\oM$ and $\Gamma^{\oM}_0$
are normally isomorphic.
\end{prop}
\begin{proof}
Let $(\cK , \pi , V)$ be a minimal Stinespring representation 
of $\gamma^\oM$
and let $\widehat{\oP} (E) := \pi (\chi_E)$
$(E \in \Sigma) .$

We first show that the mapping
\[
	\pi_1
	\colon
	L^{\infty} (\mu)
	\ni
	[f]_\mu
	\mapsto
	\pi (f) =
	\int_{\Omega}
	f (\omega)
	d\widehat{\oP} (\omega)
	\in \LK
\]
is a well-defined faithful representation.
For this it is sufficient to prove that 
$\mu (E) = 0$ if and only if
$\widehat{\oP} (E) = 0$
for
$E \in \Sigma .$
This can be shown as follows:
\begin{align*}
	\mu(E) =0
	&\iff
	\oM(E) = 0
	\\
	&\iff
	V^\ast \widehat{\oP}(E) V = 0
	\\
	&\iff
	\widehat{\oP}(E) V = 0
	\\
	&\iff
	\widehat{\oP}(E) 
	\widehat{\oP}(\Sigma)
	V \Hin 
	= \{ 0 \}
	\\
	&\iff
	\widehat{\oP}(E) 
	=0.
\end{align*}
Here, the last equivalence follows from the minimality of the Stinespring representation
$(\cK , \pi , V) .$


We next show that $\pi_1$ is normal.
For this, it is sufficient to show that 
for each family of mutually orthogonal projections
$([\chi_{E_i}]_\mu)_{i \in I}$ in $L^\infty (\mu) $
$(E_i \in \Sigma) ,$
\begin{equation}
\pi_1 \Bigl( \sum_{i \in I}   [\chi_{E_i}]_\mu  \Bigr)
=
\sum_{i \in I} 
\pi_1 (   [\chi_{E_i}]_\mu  )
=
\sum_{i \in I} 
\widehat{\oP} ( E_i ) .
\label{eq:normality}
\end{equation}
From the $\sigma$-finiteness of $L^\infty (\mu) ,$
we may assume that $I$ is countable, 
and therefore that the family
$\{ E_i \}_{i \in I}$ is mutually disjoint.
Then equation~\eqref{eq:normality} is immediate from the
countable additivity of $\widehat{\oP} .$

From the definitions of $\Gamma^\oM_0$ and $\pi_1 ,$
we have $\Gamma^\oM_0 = \Gamma^\oM \circ \pi_1 .$
By noting that $\pi_1 (L^\infty (\mu))$ is a von Neumann algebra,
we have
$\pi (B (\Omega , \Sigma )) =  \pi_1 (L^\infty (\mu)) =  
\pi (B (\Omega , \Sigma ))^{\prime \prime} .
$
Therefore $\Gamma^\oM$ and $\Gamma^\oM_0$ are normally isomorphic.
\end{proof}

The following lemma is a generalization of 
Proposition~3 in Ref.~\onlinecite{kuramochi2017minimal}.

\begin{lemm}
\label{lemm:abelout}
Let $\A$ be a commutative von Neumann algebra
and let $\Lambda \in \nchset{\A}{\LHin}$
be a normal channel.
Then there exists a POVM $\oM$ on $\Hin$
such that $\Lambda \eqncp \Gamma^\oM .$
\end{lemm}
\begin{proof}
By using the construction given in Ref.~\onlinecite{kuramochi2017minimal} (Lemma~2),
we may assume that $\Lambda$ is faithful.
Since $\A$ is commutative,
there exists a localizable measure space
$(\Omega , \Sigma , \mu)$
such that $\A$ is normally isomorphic to $L^\infty (\mu) ,$
from which we may regard $\Lambda$ as a normal channel in
$\nchset{L^\infty (\mu)}{\LHin} .$
We define a POVM $(\Omega , \Sigma , \oM)$
by
$\oM(E) := \Lambda ([\chi_E]_\mu) $
$(E \in \Sigma) .$
Let $(\cK , \pi ,V)$ be a minimal Stinespring representation of $\Lambda  .$
From the faithfulness of $\Lambda ,$
$\pi$ is a normal isomorphism from $L^\infty (\mu)$
onto $\pi (L^\infty (\mu)) .$
If we define $\pi_0 \colon B (\Omega , \Sigma ) \to \LK$ by
$\pi_0 (f) :=  \pi (  [f]_\mu)$
$(f \in B (\Omega , \Sigma )) ,$
then $(\cK , \pi_0 , V)$ is a minimal Stinespring representation of
the pre-QC channel $\gamma^\oM  .$
Thus the QC channel $\Gamma^\oM \in \nchset{ \pi(L^\infty (\mu))  }{\LHin}$
is given by
$\Gamma^{\oM} (A) = V^\ast A V $
$( A \in \pi(L^\infty (\mu)) ) .$
Then we have $\Lambda = \Gamma^\oM \circ \pi .$
Therefore we obtain $\Lambda \cong \Gamma^\oM ,$
which completes the proof.
\end{proof}


Next we define the post-processing relation for POVMs
by the concatenation relation between the corresponding QC channels.

\begin{defi}
\label{defi:postp}
Let $(\Omega_1 , \Sigma_1 , \oM)$
and
$(\Omega_2 , \Sigma_2 , \oN)$
be POVMs on $\Hin .$
\begin{enumerate}
\item
$\oM$ is a \emph{post-processing} of $\oN ,$
written as $\oM \postp \oN ,$
if $\Gamma^{\oM} \concp \Gamma^{\oN} .$
\item
$\oM$ is post-processing equivalent to $\oN ,$
written as
$\oM  \posteq \oN ,$
if both $\oM \postp \oN$ and $\oN \postp \oM$ hold.
\end{enumerate}
\end{defi}

If $\Hin$ is separable,
it is known~\cite{kuramochi2017minimal} that
$\oM \postp \oN$ 
if and only if 
there exists a mapping
$\kappa \colon \Sigma_1 \times \Omega_2 \to [0,1]$
such that
\begin{enumerate}[(i)]
\item
$\kappa(E|\cdot)$ is $\Sigma_2$-measurable
for each $E \in \Sigma_1 ;$
\item
$\kappa (\Omega_1 | \omega_2) =1 ,$
$\oN(\omega_2)$-a.e.;
\item
for each disjoint sequence 
$\{ E_n \} \subseteq \Sigma_1 ,$
$\kappa (\cup_n E_n|\omega_2) = \sum_n \kappa (E_n|\omega_2) ,$
$\oN(\omega_2)$-a.e.;
\item
$\displaystyle
\oM(E)
=
\int_{\Omega_2}
\kappa (E|\omega_2)
d\oN (\omega_2)
$
for each $E \in \Sigma_1 .$
\end{enumerate}
A mapping $\kappa$ satisfying the conditions (i)-(iii) 
is called a weak Markov kernel~\cite{jencova2008}.
In Refs.~\onlinecite{Dorofeev1997349,jencova2008},
post-processing relations for POVMs are defined by weak Markov kernel.

If $\Hin$ is non-separable,
the post-processing relation in Definition~\ref{defi:postp}
is no longer consistent with the relation
defined by weak Markov kernel 
as shown in the following example.

\begin{exam}
\label{ex:nonseparable}
Let $\Hin$ be non-separable
and let $(\phi_\omega)_{\omega \in \Omega}$
be an uncountable orthonormal basis of $\Hin .$
We define PVMs 
$(\Omega, 2^\Omega , \oM_1)$
and 
$(\Omega, \Sigma , \oM_2)$
as follows:
\begin{gather*}
	2^\Omega\text{ is the power set of $\Omega;$}
	\\
	\oM_1 (E)
	:=
	\sum_{\omega \in E}
	\ket{\phi_\omega}
	\bra{\phi_\omega} ,
	\quad
	(E \in 2^\Omega) ;
	\\
	\Sigma
	:=
	\set{E \subseteq \Omega|
	\text{either $E$ or $\Omega \setminus E$ is countable}
	};
	\\
	\oM_2 (E)
	:=
	\oM_1 (E)
	,
	\quad
	(E \in \Sigma).
\end{gather*}
Then both of the QC channels
$\Gamma^{\oM_1}$ and $\Gamma^{\oM_2}$
are normally isomorphic to the following channel:
\[
	\Gamma
	\colon
	\ell^\infty (\Omega)
	\ni
	f
	\mapsto
	\sum_{\omega \in \Omega}
	f (\omega)
	\ket{\phi_\omega}
	\bra{\phi_\omega} 
	\in \LHin ,
\]
where 
$\ell^\infty (\Omega)$
is the set of bounded complex-valued functions on $\Omega .$
Thus we have $\oM_1 \posteq \oM_2 $
in the sense of Definition~\ref{defi:postp}.
On the other hand, we can show that there exists no weak Markov kernel $\kappa$
satisfying
\begin{equation}
	\oM_1 (E)
	=
	\int_\Omega 
	\kappa (E|\omega) d \oM_2 (\omega)
	=
	\sum_{\omega \in \Omega}
	\kappa (E|\omega)
	\ket{\phi_\omega}
	\bra{\phi_\omega} ,
	\quad
	(\forall E \in 2^\Omega) .
	\label{eq:impossible}
\end{equation}
Suppose there exists such $\kappa .$
Then from equation~\eqref{eq:impossible}, we should have
$\kappa(E|\omega) = \chi_E (\omega) $
$(E \in 2^\Omega) .$
If we take $E \in 2^\Omega \setminus \Sigma \neq \varnothing ,$
$\kappa(E|\cdot) = \chi_E (\cdot)$ is not $\Sigma$-measurable,
which contradicts the definition of weak Markov kernel.
Thus $\oM_1$ and $ \oM_2$ are not equivalent
in the sense of weak Markov kernel.
\end{exam}

In the rest of this section, we will see that 
the post-processing relation in Definition~\ref{defi:postp}
is a natural relation for POVMs as long as
we consider the compatibility of POVMs and channels.

\subsection{Compatibility of POVM and channel}
Now we consider compatibility of a POVM and a normal channel.

Let $\M$ and $\Min$ be von Neumann algebras
and let $(\Omega , \Sigma)$ be a measurable space.
An 
instrument%
~\cite{davieslewisBF01647093,:/content/aip/journal/jmp/25/1/10.1063/1.526000,1751-8121-46-2-025302,1751-8121-46-2-025303} 
is a mapping 
$\I \colon \M \times \Sigma \to \Min $ such that
\begin{enumerate}[(i)]
\item
$\M \ni A \mapsto \I (A , E) \in \Min$
is a normal CP linear map for each $E \in \Sigma ;$
\item
$\Sigma \ni E \mapsto \braket{\vph , \I (A , E)} \in \cmplx$
is a countably additive complex-valued measure
for each $A \in \M $ and each $\vph \in {\Min}_\ast ;$
\item
$\I (\1_\M , \Omega) = \1_{\Min} .$
\end{enumerate}
For an instrument 
$\I \colon \M \times \Sigma \to \Min  ,$
we define the marginal POVM 
$(\Omega , \Sigma , \oM^\I)$
and the marginal channel
$\Lambda^\I \in \nchset{\M}{\Min}$
of $\I$ by
\begin{gather*}
	\oM^\I (E)
	:=
	\I (\1_\M , E) ,
	\quad
	(E \in \Sigma);
	\\
	\Lambda^\I (A)
	:=
	\I (A , \Omega )
	,
	\quad
	(A \in \M) .
\end{gather*}

\begin{defi}[Compatibility of a POVM and a channel]
\label{defi:povmchcomp}
Let $(\Omega , \Sigma , \oM)$ be a POVM on $\Hin ,$
let $\M$ be a von Neumann algebra,
and let $\Lambda \in \nchset{\M}{\LHin}$
be a normal channel.
Then $\oM$ and $\Lambda$ are said to be 
\emph{compatible},
written as $\oM \comp \Lambda ,$
if there exists an instrument
$\I \colon \M \times \Sigma \to \LHin$
such that $ \oM= \oM^\I$
and 
$\Lambda = \Lambda^\I .$
Such an instrument $\I$ is called a joint instrument
of $\oM$ and $\Lambda .$
\end{defi}

\begin{theo}
\label{theo:povmch}
Let $(\Omega , \Sigma , \oM) ,\M ,$
and
$\Lambda \in \nchset{\M}{\LHin}$
be the same as in Definition~\ref{defi:povmchcomp}.
Then the following conditions are equivalent.
\begin{enumerate}[(i)]
\item
$\oM \comp \Lambda ;$
\item
$\Gamma^{\oM} \maxcomp \Lambda ;$
\item
$\Gamma^{\oM} \mincomp \Lambda ;$
\item
$\Lambda \concp (\Gamma^\oM )^c ;$
\item
$\Gamma^\oM \concp \Lambda^c .$
\end{enumerate}
\end{theo}
\begin{proof}
The equivalence
(ii)$\iff$(iv)$\iff$(v)
follows from Theorem~\ref{theo:univbmax},
and the equivalence 
(ii)$\iff$(iii)
from that the outcome space of $\Gamma^\oM$ is nuclear.

(i)$\implies$(v).
Assume (i) and take a joint instrument
$\I \colon \M \times \Sigma \to \LHin$
of $\oM$ and $\Lambda .$
From Remark~\ref{rem:qcch}
and Lemma~\ref{lemm:minex}, we have only to show 
$\gamma^\oM \cocp \Lambda^c , $
where $\gamma^\oM \in \chset{B(\Omega ,\Sigma)}{\LHin}$
is the pre-QC channel of $\oM .$
Let $(\cK , \pi , V)$ be a minimal Stinespring representation 
of $\Lambda = \Lambda^\I .$
From the Radon-Nikodym theorem for CP maps~\cite{doi:10.1063/1.1615697},
for each $E \in \Sigma$
there exists a unique positive operator $\widehat{\oM} (E) \in \pi (\M)^\prime$
such that
$\I (A , E) = V^\ast \pi (A) \widehat{\oM} (E) V $
for all $A \in \M .$
Now we show that $(\Omega , \Sigma , \widehat{\oM})$
is a POVM on $\cK .$
For each disjoint countable family $\{ E_n \} \subseteq \Sigma ,$
each $A_1 , A_2 \in \M ,$
and each $\psi_1 , \psi_2 \in \Hin,$
we have 
\begin{align*}
	\braket{\pi(A_1)V \psi_1 | \widehat{\oM} (\cup_n E_n ) \pi (A_2) V \psi_2  }
	&=
	\braket{\psi_1 | V^\ast \pi(A_1^\ast A_2) \widehat{\oM} (\cup_n E_n )  V \psi_2  }
	\\
	&=
	\braket{\psi_1 | \I (A_1^\ast A_2 ,  \cup_n E_n )  \psi_2  }
	\\
	&=
	\sum_n
	\braket{\psi_1 | \I (A_1^\ast A_2 ,   E_n )  \psi_2  }
	\\
	&=
	\sum_n
	\braket{\pi(A_1)V \psi_1 | \widehat{\oM} (E_n ) \pi (A_2) V \psi_2  }.
\end{align*}
Thus from the minimality of the Stinespring representation $(\cK , \pi , V) ,$
$\widehat{\oM}$ is countably additive.
Since $\widehat{\oM} (\Omega) = \1_\cK$ is immediate from the definition,
$\widehat{\oM}$ is a POVM.
Therefore we can define a channel $\alpha \in \chset{B(\Omega , \Sigma)}{\pi (\M)^\prime}$
by
\[
	\alpha (f)
	:=
	\int_\Omega f (\omega) d \widehat{\oM} (\omega) ,
	\quad
	(f\in B (\Omega , \Sigma)) .
\]
Then we have 
$\gamma^\oM (\chi_E) = \oM(E) = V^\ast \widehat{\oM}(E) V = \Lambda^c \circ \alpha (\chi_E)$
$(E \in \Sigma),$
where $\Lambda^c \in \nchset{\pi(\M)^\prime}{\LHin} .$
Since the set of simple functions is norm dense in $B(\Omega , \Sigma) ,$
this implies
$\gamma^\oM = \Lambda^c \circ \alpha \cocp \Lambda^c ,$
which proves (v).

(ii)$\implies$(i).
Assume (ii).
Let $(\cK_1 , \pi_1 , V_1)$ be a minimal Stinespring representation
of $\gamma^{\oM} \in \chset{B(\Omega , \Sigma)}{\LHin} .$
Then 
$\Gamma^{\oM} \in \nchset{\pi_1 (B(\Omega , \Sigma))^{\prime \prime} }{\LHin}$
is the mapping defined by 
$\Gamma^\oM (B) = V_1^\ast B V_1 $
$(B \in \pi_1 (B(\Omega , \Sigma))^{\prime \prime} ) .$
From the assumption,
we can take a binormal $\max$-joint channel
$\Theta \in \chset{\M \maxtensor \pi_1 (B(\Omega , \Sigma))^{\prime \prime} }{\LHin}$
of  
$ \Lambda $
and $\Gamma^\oM .$
We define a mapping
$\I \colon \M \times \Sigma \to \LHin$
by
$\I (A , E ) := \Theta (A \otimes \pi_1 (\chi_E)) $
$(A \in \M , E \in \Sigma) .$
Then $\I$ is a joint instrument
of $\oM$ and $\Lambda ,$
which proves $\oM \comp \Lambda .$
\end{proof}

Theorem~\ref{theo:povmch} indicates that
for a POVM $\oM ,$ the conjugate channel
$(\Gamma^\oM)^c$ can be interpreted as 
the least-disturbing channel compatible with $\oM ,$
generalizing the result obtained in Ref.~\onlinecite{PhysRevA.88.042117}
for discrete POVMs.

\begin{rem}
\label{rem:theo4}
In Ref.~\onlinecite{1751-8121-46-2-025303}, the characterization was given 
for all \emph{instruments} that have a given POVM $\oM$ as their marginal POVM
(Theorem~1 and Corollary~1) by using direct-integral decomposition.~\cite{1751-8121-46-2-025302}
The channel $T$ appearing in Theorem~1 of 
Ref.~\onlinecite{1751-8121-46-2-025303}
satisfies, in the setting of our Theorem~\ref{theo:povmch},  
$\Lambda = (\Gamma^\oM)^c \circ T ,$
which corresponds to the condition (iv) of Theorem~\ref{theo:povmch}.
\end{rem}

\subsection{Compatibility of two POVMs}
\begin{defi}
\label{defi:povmscomp}
Let $(\Omega_1 , \Sigma_1 , \oM_1 )$
and $(\Omega_2 , \Sigma_2 , \oM_2)$
be POVMs on $\Hin .$
$\oM_1$ and $\oM_2$ are said to be \emph{compatible},
written as $\oM_1 \comp \oM_2 ,$
if there exists a POVM $\oN$ on $\Hin$
such that
$\oM_1 \postp \oN$ and $\oM_2 \postp \oN .$
\end{defi}
\begin{rem}
\label{rem:povmc}
If $\Hin$ is separable and the outcome spaces of $\oM_1$
and $\oM_2$ are standard Borel~\cite{kechris1995classical},
$\oM_1 \comp \oM_2$
if and only if
$\oM_1 $ and $ \oM_2$
are jointly measurable~\cite{busch2016quantum},
i.e.\ there exists a POVM
$(\Omega_1 \times \Omega_2 , \Sigma_1 \otimes \Sigma_2 , \oM_{12})$
such that
$\oM_{12}(E_1 \times \Omega_2) = \oM_1 (E_1)$
and
$\oM_{12}(\Omega_1 \times E_2) = \oM_2 (E_2)$
$(E_1 \in \Sigma_1 , E_2 \in \Sigma_2) .$
Here $\Sigma_1 \otimes \Sigma_2$ denotes 
the product $\sigma$-algebra of $\Sigma_1$ and $\Sigma_2 .$
\end{rem}

The compatibility of POVMs in the above definition is 
consistent with the compatibility of channels 
as shown in the following proposition.

\begin{prop}
\label{prop:comppovms}
Let $(\Omega_1 , \Sigma_1 , \oM_1 )$
and $(\Omega_2 , \Sigma_2 , \oM_2)$
be POVMs on $\Hin .$
Then the following conditions are equivalent.
\begin{enumerate}[(i)]
\item
$\oM_1 \comp \oM_2 ;$ 
\item
$\Gamma^{\oM_1} \maxcomp \Gamma^{\oM_2} ;$
\item
$\gamma^{\oM_1} \maxcomp \gamma^{\oM_2} ;$
\item
there exist Naimark dilations 
$(\cK , \oP_1 , V)$ 
and
$(\cK , \oP_2, V)$
of $\oM_1$ and $\oM_2$, respectively,
such that the ranges
$\oP_1 (\Sigma_1)$ 
and 
$\oP_2 (\Sigma_2)$ 
commute.
\end{enumerate}
\end{prop}
\begin{proof}
(i)$\implies$(ii).
Assume $\oM_1 \comp \oM_2 .$
Then there exists a POVM
$(\Omega_0 , \Sigma_0 , \oN)$
on $\Hin$ satisfying $\Gamma^{\oM_1} \concp \Gamma^\oN $
and $\Gamma^{\oM_2} \concp \Gamma^\oN .$
We write the outcome spaces of $\Gamma^{\oM_1} , \Gamma^{\oM_2}$ and
$\Gamma^\oN$ as
$\A_{\oM_1} ,\A_{\oM_2} ,$ and $\A_{\oN} ,$
respectively.
By assumption, we can take normal channels
$\alpha_1 \in \nchset{\A_{\oM_1}}{\A_\oN}$
and
$\alpha_2 \in \nchset{\A_{\oM_2}}{\A_\oN}$
such that
$\Gamma^{\oM_1} = \Gamma^\oN \circ \alpha_1$
and
$\Gamma^{\oM_2} = \Gamma^\oN \circ \alpha_2 .$
Since $\A_\oN$ is commutative,
there exists a representation 
$\tilde{\pi} \colon  \A_\oN \maxtensor \A_\oN \to \A_\oN$
such that $\tilde{\pi} (B_1 \otimes B_2) = B_1 B_2 $
$(B_1 , B_2 \in \A_\oN) .$
If we define $\Theta \in \chset{ \A_{\oM_1} \maxtensor \A_{\oM_2} }{\LHin} $
by
$\Theta := \Gamma^\oN \circ \tilde{\pi} \circ (\alpha_1 \maxtensor \alpha_2),$
$\Theta$ is a $\max$-joint channel of $\Gamma^{\oM_1}$ and $\Gamma^{\oM_2} ,$
which implies $\Gamma^{\oM_1} \maxcomp \Gamma^{\oM_2} .$

(ii)$\implies$(iii) is immediate from 
$\gamma^{\oM_1} \cocp \Gamma^{\oM_1}$
and
$\gamma^{\oM_2} \cocp \Gamma^{\oM_2} .$

(iii)$\implies$(iv).
Assume $\gamma^{\oM_1} \maxcomp \gamma^{\oM_2}  .$
Then we can take a $\max$-joint channel
$
\Theta \in 
\chset{
B(\Omega_1 , \Sigma_1) \maxtensor B(\Omega_2 , \Sigma_2)
}{
\LHin
}
$
of 
$\gamma^{\oM_1} $ and $ \gamma^{\oM_2} .$
Let $(\cK_\Theta , \pi_\Theta , V_\Theta)$ be
a minimal Stinespring representation of $\Theta .$
Then for each $E_1 \in \Sigma_1$ and $E_2 \in \Sigma_2$ we have
\begin{gather*}
	\oM_1 (E_1)
	=
	\gamma^{\oM_1}
	(\chi_{E_1})
	=
	\Theta (\chi_{E_1} \otimes \chi_{\Omega_2})
	=
	V_\Theta^\ast 
	\oP_1 (E_1) 
	V_\Theta , \\
	\oM_2 (E_2)
	=
	\gamma^{\oM_2}
	(\chi_{E_2})
	=
	\Theta (\chi_{\Omega_1} \otimes \chi_{E_2})
	=
	V_\Theta^\ast 
	\oP_2 (E_2) 
	V_\Theta , 
\end{gather*}
where
$\oP_1 (E_1) := \pi_\Theta (\chi_{E_1} \otimes \chi_{\Omega_2})$
and
$\oP_2 (E_2) := \pi_\Theta (\chi_{\Omega_1} \otimes \chi_{E_2}) .$
By a similar discussion as in Lemma~\ref{lemm:normalrep}, we can show that
$(\Omega_1 , \Sigma_1 , \oP_1)$ 
and 
$(\Omega_2 , \Sigma_2 , \oP_2)$ 
are PVMs on $\cK_\Theta .$
Thus $(\cK_\Theta , \oP_1 , V_\Theta )$
and $(\cK_\Theta , \oP_2 , V_\Theta )$
are Nairmark dilations of $\oM_1$ and $\oM_2 ,$
respectively, with commuting ranges,
which proves (iv).

(iv)$\implies$(i).
Let $(\cK , \oP_1 , V)$ and $(\cK ,\oP_2 , V)$
be Naimark dilations of $\oM_1$ and $\oM_2 ,$
respectively, with commuting ranges.
We define representations
$\pi_1$ and $\pi_2$ by
\begin{gather*}
	\pi_1
	\colon
	B(\Omega_1 , \Sigma_1)
	\ni f_1
	\mapsto
	\int_{\Omega_1} 
	f_1 (\omega_1)
	d \oP_1 (\omega_1) 
	\in \LK,
	\\
	\pi_2
	\colon
	B(\Omega_2 , \Sigma_2)
	\ni f_2
	\mapsto
	\int_{\Omega_2} 
	f_2 (\omega_2)
	d \oP_2 (\omega_2) 
	\in \LK,
\end{gather*}
a commutative von Neumann algebra
$\widetilde{\A}$ by 
$( \oP_1 (\Sigma_1) \cup \oP_2 (\Sigma_2) )^{\prime \prime} ,$
and a normal channel
$\widetilde{\Lambda} \in \nchset{\widetilde{\A}}{\LHin} $
by
$\widetilde{\Lambda} (X) := V^\ast X V$
$(X \in \widetilde{\A}) .$
Since the ranges of $\pi_1$ and $\pi_2$
are contained in $\widetilde{\A} ,$
for each $E_j \in \Sigma_j$
$(j=1,2)$
we have
\[
\gamma^{\oM_j} (\chi_{E_j})
=
\oM_j (E_j)
=
V^\ast \oP_j (E_j ) V
=
\widetilde{\Lambda} 
\circ \pi_j
(\chi_{E_j}) ,
\]
which implies 
$
\gamma^{\oM_j} 
=
\widetilde{\Lambda} 
\circ \pi_j
\cocp \widetilde{\Lambda} .$
Thus from Remark~\ref{rem:qcch} and Lemma~\ref{lemm:minex},
we have $\Gamma^{\oM_j} \concp \widetilde{\Lambda} $
$(j=1,2) .$
Since $\widetilde{\A}$ is commutative,
from Lemma~\ref{lemm:abelout} we can take 
a POVM $\oN$ on $\Hin$ satisfying
$\widetilde{\Lambda} \eqncp \Gamma^\oN .$
Thus we obtain $\oM_j \postp \oN$
$(j=1,2) ,$
which implies the condition~(i).
\end{proof}

\begin{rem}
\label{rem:beneduci}
The equivalence of the condition~(iv) in Proposition~\ref{prop:comppovms}
and the joint measurability of $\oM_1$ and $\oM_2$
was shown in Ref.~\onlinecite{Beneduci2017197} (Theorem~6)
under some weak assumptions on the outcome measurable spaces.
\end{rem}

\subsection{Maximal POVM}
A POVM $\oM$ on $\Hin$ is called 
maximal~\cite{Dorofeev1997349,ISI:000236487200002}
if $\oM \postp \oN$ implies $\oN \postp \oM$
for any POVM $\oN$ on $\Hin .$
The following theorem gives a characterization of the maximality
of a POVM in terms of the conjugate channel of the QC channel.

\begin{theo}
\label{theo:maximalpovm}
Let $(\Omega , \Sigma , \oM)$ be a POVM on $\Hin .$
Then the following conditions are equivalent.
\begin{enumerate}[(i)]
\item
$\oM$ is maximal;
\item
there exist a commutative von Neumann algebra $\A_0$
and a normal channel $\Gamma_0 \in \nchset{\A_0}{\LHin}$
with a minimal Stinespring representation $(\cK_0 , \pi_0 , V_0)$
such that
$\Gamma^\oM \eqncp \Gamma_0$
and $\pi_0(\A_0)^\prime = \pi_0 (\A_0)$;
\item
$\Gamma^\oM \eqncp (\Gamma^\oM)^c .$
\end{enumerate}
\end{theo}
\begin{proof}
(i)$\implies$(ii). 
Assume (i).
From Ref.~\onlinecite{kuramochi2017minimal},
there exist a commutative von Neumann algebra $\A_0 $
and normal channel $\Gamma_0 \in \nchset{\A_0}{\LHin}$
such that $\Gamma_0 \eqncp \Gamma^\oM$ and
$\Gamma_0$ is minimal sufficient.
Let $(\cK_0 , \pi_0 , V_0)$ be a minimal Stinespring representation of 
$\Gamma_0$
and let $\widetilde{\Gamma}_0 \in \nchset{\pi_0 (\A_0)}{\LHin}$
be the normal channel defined by
$\widetilde{\Gamma}_0 (A) := V_0^\ast A V_0 $
$(A \in \pi_0 (\A_0)) .$
Since $\pi_0$ is a normal isomorphism from $\A_0$
onto $\pi_0 (\A_0) $
due to the faithfulness of $\Gamma_0 ,$
$\Gamma_0$ and $\widetilde{\Gamma}_0$ are normally isomorphic,
and therefore $\widetilde{\Gamma}_0$ is also minimal sufficient.
Now we show $\pi_0 (\A_0)^\prime = \pi_0 (\A_0) ,$
which implies the condition~(ii).
Take an arbitrary self-adjoint element $B \in \pi_0 (\A_0)^\prime$
and let $\widetilde{\A}_B$ denote the commutative von Neumann algebra 
generated by $\pi_0 (\A_0) \cup \{ B \} .$
We define a normal channel 
$\Lambda_B \in \nchset{\widetilde{\A}_B}{\LHin}$
by
$\Lambda_B (C) := V_0^\ast C V_0$
$(C \in \widetilde{\A}_B ) .$
Since $\pi_0 (\A_0) \subseteq \widetilde{\A}_B ,$
we have $\Gamma^\oM \eqncp \widetilde{\Gamma}_0 \concp \Lambda_B .$
On the other hand, 
since the outcome space of $\Lambda_B$ is commutative,
Lemma~\ref{lemm:abelout} implies $\Lambda_B \eqncp \Gamma^\oN$
for some POVM $\oN$ on $\Hin .$
Thus from the maximality of $\oM ,$
we obtain $\Lambda_B \concp \widetilde{\Gamma}_0 .$
Hence there exists a normal channel 
$\E \in \nchset{\widetilde{\A}_B}{\pi_0 (\A_0)}$
such that $\Lambda_B =  \widetilde{\Gamma}_0 \circ \E .$
Then we have 
$\widetilde{\Gamma}_0 (A) = \Lambda_B (A) = \widetilde{\Gamma}_0 \circ \E (A) $
for all $A \in \pi_0 (\A_0) .$
From the minimal sufficiency of $\widetilde{\Gamma}_0 ,$
this implies $\E (A) = A $ for all $A \in \pi_0 (\A_0) .$
Therefore $\E$ is a normal norm-$1$ projection from $\widetilde{\A}_B$ onto 
the subalgebra $\pi_0 (\A_0) .$
From $\Lambda_B = \widetilde{\Gamma}_0 \circ \E ,$
we obtain
$
	V_0^\ast C V_0 
	=
	V_0^\ast \E ( C ) V_0 
$
$(C \in \widetilde{\A}_B) .$
Hence, for each $A_1 , A_2 \in \pi_0 (\A_0)$
and
each $\psi_1 , \psi_2 \in \Hin ,$
we have
\begin{align} 
	\braket{A_1 V_0 \psi_1 | \E (B) A_2 V_0 \psi_2}
	&=
	\braket{\psi_1 |V_0^\ast A_1^\ast \E (B) A_2 V_0 \psi_2}
	\notag
	\\
	&=
	\braket{\psi_1 |V_0^\ast \E (  A_1^\ast B A_2 )  V_0 \psi_2}
	\label{eq:tomiyama}
	\\
	&=
	\braket{\psi_1 |V_0^\ast   A_1^\ast B A_2   V_0 \psi_2}
	\notag
	\\
	&=
	\braket{A_1 V_0 \psi_1 | B A_2 V_0 \psi_2} ,
	\notag
\end{align}
where we have used Tomiyama's theorem
(e.g.\ Ref.~\onlinecite{brown2008c}, Theorem~1.5.10)
in deriving the equality~\eqref{eq:tomiyama}.
From the minimality of the Stinespring representation
$(\cK_0 , \pi_0 , V_0),$
this implies 
$B = \E (B)  \in \pi_0 (\A_0) .$
Therefore we obtain $\pi_0 (\A_0) = \pi_0 (\A_0)^\prime .$

(ii)$\implies$(iii).
Assume (ii). Then from the definition of 
the commutant conjugate channel and Corollary~\ref{coro:cc},
we have 
$\Gamma_0 \eqncp \Gamma_0^{c} .$
Since we have $(\Gamma^\oM )^c \eqncp \Gamma_0^c$
from Theorem~\ref{theo:equiv},
we obtain $\Gamma^\oM \eqncp (\Gamma^\oM)^c .$

(iii)$\implies$(i).
Assume (iii) and take an arbitrary POVM $\oN$ on
$\Hin$ satisfying
$\oM \postp \oN .$
Then from 
Lemma~\ref{lemm:qcconj}, 
Theorem~\ref{theo:equiv},
and $\Gamma^\oM \concp \Gamma^\oN ,$
we have
$\Gamma^{\oN} \concp (\Gamma^\oN)^c \concp (\Gamma^\oM)^c \eqncp \Gamma^\oM ,$
which implies $\oN \postp \oM .$
Therefore $\oM$ is maximal.
\end{proof}
\begin{rem}
\label{rem:maximal}
In Ref.~\onlinecite{1751-8121-46-2-025303} (Section~4),
the concept of rank-$1$ POVM is introduced by using direct integrals.
Then, for separable $\Hin ,$ the equivalence (i)$\iff$(ii) in 
Theorem~\ref{theo:maximalpovm}
can be rephrased as
\lq\lq{}$\oM$ is maximal if and only if $\oM$ is of rank-$1$\rq\rq{},
which was first obtained in Ref.~\onlinecite{ISI:000236487200002}.
The equivalence (i)$\iff$(iii) in Theorem~\ref{theo:maximalpovm}
for discrete $\oM$ was first proved by Heinosaari and Miyadera
(T. Heinosaari and T. Miyadera, private communications).
\end{rem}
The following corollary was essentially obtained in Ref.~\onlinecite{1751-8121-46-2-025303} (Section~4) for separable $\Hin .$
\begin{coro}
\label{coro:maximalub}
For any POVM $(\Omega , \Sigma , \oM)$ on $\Hin ,$
there exists a maximal POVM $\oN$ on $\Hin$
satisfying $\oM \postp \oN .$
\end{coro}
\begin{proof}
Let 
$\A_\oM$ be the outcome space of the QC channel $\Gamma^\oM$
and let
$(\cK  , \pi , V )$
be a minimal Stinespring representation of $\Gamma^\oM .$
Then, by Zorn's lemma, 
$\pi (\A_\oM)$ is contained in 
a commutative von Neumann algebra 
$  \A_0  $ on $\cK $
maximal with respect to set inclusion $\subseteq .$
From the maximality of $\A_0$ we have
$\A_0^\prime = \A_0 .$
Now we define $\Gamma_0 \in \nchset{\A_0 }{\LHin}$
by
$\Gamma_0 (B) := V^\ast B V$
$(B \in \A_0)$
and take a POVM $\oN$
on $\Hin$ satisfying $\Gamma^\oN \eqncp \Gamma_0 .$
Then $\oN$ satisfies the condition~(ii) of Theorem~\ref{theo:maximalpovm},
and therefore is a maximal POVM.
Furthermore, we have $\Gamma^\oM \concp \Gamma_0 \eqncp \Gamma^\oN$
from the definition of $\Gamma_0 .$
Thus $\oM \postp \oN .$
\end{proof}

\section{Normal compatibility and tensor conjugate channel} \label{sec:normal}
In this section, we consider the normal compatibility relation $\ncomp .$

\begin{defi}
\label{defi:tensorst}
Let $\M$ be a von Neumann algebra acting on a Hilbert space $\Hout$, 
let $\Hin$ be a Hilbert space,
and let $\Lambda \in \nchset{ \M }{\LHin}$
be a normal channel.
A pair $(\cK , V)$ is called a \emph{tensor Stinespring representation} of $\Lambda$
if $\cK$ is a Hilbert space and $V \colon \Hin \to \Hout \otimes \cK$ is an isometry
such that
$
	\Lambda (A)
	=
	V^\ast
	(A \otimes \1_\cK)
	V 
$
for all
$ A \in \M  .$
\end{defi}
A tensor Stinespring representations $(\cK , V)$ can be regarded as 
a special Stinespring representation 
$(\Hout \otimes \cK , \pi_\cK , V) ,$
where
$\pi_\cK (A) := A \otimes \1_\cK .$

\begin{defi}
\label{defi:conj}
Let 
$\M,$ 
$\Hout ,$ 
$\Hin , $
and $\Lambda \in \nchset{ \M }{\LHin}$
be the same as in Definition~\ref{defi:tensorst}.
For a tensor Stinespring representation $(\cK , V)$ of $\Lambda ,$
the channel $\Gamma \in \chset{ \LK }{ \LHin}$ 
defined by
$
\Gamma(B)
:=
V^\ast( \1_{\Hout} \otimes B) V
$
$(B \in \LK)$
is called a 
\emph{tensor conjugate channel} of $\Lambda .$
A tensor conjugate channel is, by definition,
a fully quantum channel.
\end{defi}

Since we have $\cL (\Hout \otimes \cK) = \cL (\Hout) \ntensor \cL (\cK ) ,$
the following corollary is immediate from the definitions of
the tensor conjugate channel and the normal compatibility.
\begin{coro}
\label{coro:tconjcomp}
Let $\Lambda$ be a normal channel with a fully quantum input space
and let $\Gamma $ be a tensor conjugate channel of $\Lambda.$
Then we have $\Lambda \ncomp \Gamma .$
\end{coro}

A tensor Stinespring representation always exists for any normal channel 
as the following proposition shows.

\begin{prop}
\label{prop:tensorrep}
Let 
$\M,$ 
$\Hout ,$ 
$\Hin , $
and $\Lambda \in \nchset{ \M }{\LHin}$
be the same as in Definition~\ref{defi:tensorst}.
Then there exists a tensor Stinespring representation of $\Lambda .$
\end{prop}
\begin{proof}
From Stinespring's dilation theorem, we can take a minimal 
Stinespring representation $(\cK_0 , \pi_0 , V_0 )$
of $\Lambda .$
According to Theorem~IV.5.5 of Ref.~\onlinecite{takesakivol1},
there exist a Hilbert space $\cK , $
a projection $E \in \M^\prime \ntensor \LK = (\M \ntensor \mathbb{C} \1_\cK)^\prime ,$
and an isometry 
$U \colon E (\Hout \otimes \cK) \to \cK_0$
such that
\begin{equation}
	\pi_0 (A)
	=
	U 
	(A \otimes \1_\cK)
	E
	U^\ast
	\quad
	(\forall A \in \M) .
	\notag
\end{equation}
Then
$V := EU^\ast V_0  $
is an isometry such that
$\Lambda (A) =  V^\ast (A \otimes \1_\cK) V$
$(\forall A \in \M) .$
Thus $(\cK , V)$ is a tensor Stinespring representation of $\Lambda$.
\end{proof}

Now let us consider fully quantum channels.
Let $\Lambda \in \nchset{\LHout}{\LHin}$ be a fully quantum channel
with a minimal Stinespring representation $(\cK , \pi , V) .$
Then, since $\pi$ is normal, $\cK$ and $\pi$ can be taken 
such that $\cK = \mathcal{H}_{\mathrm{out}} \otimes \cK^\prime$
and $\pi (A) = A \otimes \1_{\cK^\prime}$
(e.g.\ Ref.~\onlinecite{davies1976quantum}, Lemma~9.2.2).
Thus we have $\Lambda (A) = V^\ast (A \otimes \1_{\cK^\prime})V$
$(A \in \LHout)$ and
the concepts of minimal and tensor Stinespring representations coincide in this case. 
Furthermore, for a fully quantum channel, 
tensor and commutant conjugate channels coincide up to concatenation equivalence
as in the following proposition.
\begin{prop}
\label{prop:pq}
Let $\Hin$ and $\Hout$ be Hilbert spaces 
and let 
$\Lambda \in \nchset{ \LHout }{ \LHin  }$ 
be a fully quantum channel.
Then any tensor and commutant conjugate channels of $\Lambda$
are normally concatenation equivalent. 
\end{prop}
\begin{proof}
Let $(\cK_1 , V_1)$ be a tensor Stinespring representation of $\Lambda $
and
let $\Gamma \in \chset{ \cL (\cK_1) }{\LHin}$
be the corresponding tensor conjugate channel of $\Lambda. $
Then $(\Hout \otimes \cK_1 , \pi , V_1)$ is a Stinespring representation of $\Lambda ,$
where $\pi (A) := A \otimes \1_{\cK_1} $ 
$(A \in \LHout) ,$
and the corresponding commutant conjugate channel is the map
$\Lambda^c \colon \pi(\LHout)^\prime =  \mathbb{C} \1_{\Hout} \ntensor \cL (\cK_1) \to \LHin   $
given by
\begin{equation*}
	\Lambda^c (\1_{\Hout} \otimes B ) = V_1^\ast ( \1_{\Hout}  \otimes B  ) V_1,
	\quad
	(B \in \cL (\cK_1)) .
\end{equation*}
Then we have $\Gamma = \Lambda^c \circ \pi_1$
and $\Lambda^c = \Gamma \circ \pi_1^{-1},$
where 
\begin{equation*}
	\pi_1 \colon \LKone \ni B \mapsto 
	\1_{\Hout} \otimes B 
	\in \mathbb{C} \1_{\Hout} \ntensor \LKone
\end{equation*}
is a normal isomorphism from
$\LKone$
onto
$\mathbb{C} \1_{\Hout} \ntensor \LKone .$
Thus we have $\Lambda^c \eqncp \Gamma .$
Since Proposition~\ref{prop:cconj1} assures that 
any commutant conjugate channels of $\Lambda$ are mutually equivalent,
this proves the assertion.
\end{proof}

The normal and \cstar-compatibility relations 
for a fully quantum channel coincide as shown in the following corollary.

\begin{coro}
\label{coro:pqconj}
Let $\cH_1 $ and $\Hin$ be Hilbert spaces,
let $\N$ be a von Neumann algebra,
let $\Lambda \in \nchset{\LHone}{ \LHin}$ be a fully quantum channel,
and let $\Gamma \in \nchset{\N}{\LHin}$ be a normal channel.
Then the following conditions are equivalent.
\begin{enumerate}[(i)]
\item
$\Lambda \ncomp \Gamma ;$
\item
$\Lambda \maxcomp \Gamma ;$
\item
$\Lambda \concp \Gamma^c ;$
\item
$\Gamma \concp \Lambda^c .$
\end{enumerate}
\end{coro}
\begin{proof}
The implications 
(i)$\implies$(ii)$\iff$(iii)$\iff$(iv)
follow from Lemma~\ref{lemm:cimplication} and 
Theorem~\ref{theo:univbmax}.
(iv)$\implies$(i) follows from 
$\Lambda \ncomp \Lambda^c ,$
which is immediate from Corollary~\ref{coro:tconjcomp} and 
Proposition~\ref{prop:pq}.
\end{proof}

The following theorem states that 
for a given normal channel $\Lambda, $
the class of normally compatible channels 
is upper bounded by 
the class of tensor conjugate channels of $\Lambda .$

\begin{theo}
\label{theo:tconj}
Let $\M_1$ be a von Neumann algebra acting on a Hilbert space $\cH_1 $
and let $\Lambda_1 \in \nchset{ \M_1 }{ \LHin }$ be a normal channel.
Then for any normal channel $\Lambda_2$ 
satisfying $\Lambda_1 \ncomp \Lambda_2 ,$
there exists a tensor conjugate channel $\Gamma$ of $\Lambda_1$ such that
$\Lambda_2 \concp \Gamma.$
\end{theo}
\begin{proof}
Let $\M_2$ be the outcome space of $\Lambda_2 $ acting on a Hilbert space $\cH_2 ,$
and let $\Lambda_{12} \in \nchset{ \M_1 \ntensor \M_2  }{\LHin}$
be a normal joint channel of $\Lambda_1$ and $\Lambda_2.$
Proposition~\ref{prop:tensorrep} implies that 
there exists a tensor Stinespring representation $(\cK , V)$ of $\Lambda_{12} . $
Then for each $A \in \M_2 $ we have
\begin{align*}
	\Lambda_2 (A)
	&=
	\Lambda_{12} (\1_{\cH_1} \otimes A)
	=
	V^\ast 
	(\1_{\cH_1}  \otimes A \otimes \1_{\cK} )
	V.
\end{align*}
Thus if we define a normal channel 
$\Gamma \in \nchset{ \cL (\cH_2 \otimes \cK )  }{\LHin}$
by
$ \Gamma (B)  :=  V^\ast (\1_{\cH_1}   \otimes B)   V $
$(B \in \cL (\cH_2 \otimes \cK ) ) ,$ 
$\Gamma $ is a tensor conjugate channel of $\Lambda_1$ and satisfies
$\Lambda_2 = \Gamma \circ \Theta ,$
where 
$\Theta \in \nchset{\M_2}{ \cL (\cH_2 \otimes \cK )}$
is defined by
$\Theta (A)  :=  A \otimes \1_{\cK} $
$(A \in \M_2 ) .$
Hence we have $\Lambda_2 \concp \Gamma .$ 
\end{proof}

\begin{coro}
\label{coro:fq}
Let $\Lambda \in \nchset{\M}{\LHin}$ be a normal channel
with the outcome von Neumann algebra $\M .$
Then
$\Lambda \ncomp \Lambda^c$
if and only if
$\Lambda$ is normally concatenation equivalent to
a fully quantum channel.
\end{coro}
\begin{proof}
Assume $\Lambda \ncomp \Lambda^c .$
Then Theorem~\ref{theo:tconj} implies that
there exists a tensor conjugate channel
$\Gamma \in \nchset{\LK}{\LHin}$ of $\Lambda $
satisfying
$\Lambda^c \concp \Gamma .$
From $\Lambda \maxcomp \Gamma ,$
we have $\Gamma \eqncp \Lambda^c .$
Thus from Corollary~\ref{coro:cc}
we obtain $\Lambda \eqncp \Lambda^{cc} \eqncp \Gamma^c .$
Since $\Gamma^c$ can be taken to be fully quantum,
$\Lambda$ is equivalent to a fully quantum channel.
Conversely, if $\Lambda \eqncp \Lambda_0$ for some fully quantum channel,
we have $\Lambda_0 \ncomp \Lambda_0^c .$
Thus from Lemma~\ref{lemm:ccconsis} we obtain $\Lambda \ncomp \Lambda^c .$
\end{proof}

\section{Comparison of the compatibility relations} \label{sec:comparison}
In this section, we compare the \cstar- and normal compatibility relations
for QC channels.
For this purpose we first show the following two propositions.

\begin{prop}
\label{prop:ccomm}
Let $\A$ be a \cstar-algebra and let $\A \gtensor \A$ be a \cstar-tensor product.
Then the following conditions are equivalent.
\begin{enumerate}[(i)]
\item
$\A$ is commutative;
\item
there exits a
channel
$\Lambda \in \chset{ \A \gtensor \A}{\A}$
satisfying
$
\Lambda (A \otimes \1_\A ) = 
\Lambda (\1_\A \otimes  A ) =
A  
$
for all $A \in \A .$
\end{enumerate}
Furthermore, the map $\Lambda$ in the condition~(ii) is
the representation such that
$\Lambda (A \otimes B) = AB$
for all $A , B \in \A .$
\end{prop}
\begin{proof}
Assume~(i).
Since we have $\A \gtensor \A = \A \maxtensor \A ,$
the representation $\Lambda \colon \A \maxtensor \A \to \A$
given by
$\Lambda (A \otimes B) =AB$
$(A,B \in \A)$
satisfies the condition~(ii).

Assume~(ii).
Then we have
\begin{gather*}
	\Lambda ((A \otimes \1_\A )^\ast (A \otimes \1_\A ))
	=A^\ast A 
	=
	\Lambda ((A \otimes \1_\A )^\ast )
	\Lambda ( (A \otimes \1_\A )) ,
	\\
	\Lambda ((A \otimes \1_\A ) (A \otimes \1_\A )^\ast)
	=A A^\ast 
	=
	\Lambda ((A \otimes \1_\A ) )
	\Lambda ( (A \otimes \1_\A )^\ast) 
\end{gather*}
for all $A \in \A .$
Thus $\A \otimes \1_\A $ is contained in the multiplicative domain 
of $\Lambda .$
Therefore for each $A , B \in \A ,$ we have
\begin{align*}
	AB
	&=
	\Lambda (A \otimes \1_\A)
	\Lambda (\1_\A \otimes B)
	\\
	&=
	\Lambda (A \otimes B)
	\\
	&=
	\Lambda (\1_\A \otimes B)
	\Lambda (A \otimes \1_\A)
	\\
	&=BA .
\end{align*}
Hence $\A$ is commutative.
This also shows that $\Lambda$ satisfies the last part of the claim.
\end{proof}

A commutative von Neumann algebra $\A$ is called atomic if 
$\A$ is generated by minimal projections.
A commutative von Neumann algebra $\A$ is atomic if and only if $\A$ is normally isomorphic to
$\ell^\infty (\Omega)$ for a set $\Omega .$

\begin{prop}
\label{prop:ncomm}
Let $\A$ be a von Neumann algebra.
Then the following conditions are equivalent.
\begin{enumerate}[(i)]
\item
$\A$ is an atomic commutative von Neumann algebra;
\item
there exits a normal channel
$\Lambda \in \nchset{ \A \ntensor \A }{ \A}$
satisfying
$
\Lambda (A \otimes \1_\A ) = 
\Lambda (\1_\A \otimes  A ) =
A  
$
for all $A \in \A .$
\end{enumerate}
Furthermore, the map $\Lambda$ in the condition~(ii) is
the representation such that
$\Lambda (A \otimes B) = AB$
for all $A , B \in \A .$
\end{prop}
\begin{proof}
Assume~(ii).
Then a similar proof as in Proposition~\ref{prop:ccomm}
yields that $\A$ is commutative and 
$\Lambda$ is the normal representation satisfying
$\Lambda (A \otimes B) = AB$
for all $A , B \in \A .$
The rest of the claim follows from Proposition~3.6 of
Ref.~\onlinecite{Kaniowski2015}.
\end{proof}

Suppose that, in Proposition~\ref{prop:ccomm},
$\A$ is the commutative von Neumann algebra
$L^\infty (\mu)$ for some $\sigma$-finite measure space
$(\Omega ,\Sigma , \mu) .$
Now we show that $\Lambda$ in Proposition~\ref{prop:ccomm}
corresponds to the copying operation of 
the classical outcome $ \omega \in \Omega .$
Take an arbitrary normal state $\vph$ on $L^\infty (\mu) .$
Then there exists a unique probability measure $P$ on $(\Omega , \Sigma)$
absolutely continuous with respect to $\mu$ such that
\[
	\vph ([f]_\mu)
	=
	\int_\Omega
	f(\omega) dP(\omega)
\]
for all
$[f]_\mu \in L^\infty (\mu) .$
By the copying operation 
$
\Omega \ni \omega 
\mapsto 
(\omega , \omega) 
\in 
\Omega \times \Omega ,
$
the probability measure $P$ is transformed into the probability measure
$\widetilde{P}$ on
$(\Omega \times \Omega , \Sigma \otimes \Sigma )$
defined by
$
	\widetilde{P}(E) := 
	P (
	\set{\omega \in \Omega | (\omega , \omega ) \in E}
	)
$
for $E \in \Sigma \otimes \Sigma .$
Then for each $E_1 , E_2 \in \Sigma $ we have
\begin{align*}
	\widetilde{P}(E_1 \times E_2)
	&=
	P(E_1 \cap E_2)
	\\
	&=
	\int_\Omega
	\chi_{E_1} (\omega)
	\chi_{E_2} (\omega)
	dP (\omega)
	\\
	&=
	\braket{
	\vph ,
	[\chi_{E_1}]_\mu  [\chi_{E_2}]_\mu 
	}
	\\
	&=
	\braket{
	\vph ,
	\Lambda ( [\chi_{E_1}]_\mu \otimes  [\chi_{E_2}]_\mu ) 
	}
	\\
	&=
	\braket{
	\Lambda^\ast \vph ,
	 [\chi_{E_1}]_\mu \otimes  [\chi_{E_2}]_\mu 
	} .
\end{align*}
Therefore the state $\Lambda^\ast \vph$ corresponds to the probability distribution
$\widetilde{P}$ obtained from the copying operation.

On the other hand, Proposition~\ref{prop:ncomm} indicates that
such copying operations are possible only for discrete spaces
if we identify the composite outcome space as the normal tensor product.
We can also see this from the following observation.
Let $(\Omega , \Sigma ,\mu) , $ $\vph , $ 
$P,$ and $\widetilde{P}$ be the same as the above.
The normal tensor product 
$L^\infty (\mu) \ntensor  L^\infty (\mu)$ is normally isomorphic to
$L^\infty (\mu \otimes \mu) ,$
where $\mu \otimes \mu$ is the direct product measure on
$(\Omega\times \Omega , \Sigma \otimes \Sigma ) .$
Since the measure $\widetilde{P}$ is concentrated
on the diagonal set
$
D := 
\set{ 
(\omega_1 , \omega_2) \in \Omega \times \Omega
|
\omega_1 = \omega_2
} ,
$
$\widetilde{P}$ is not absolutely continuous
with respect to $\mu \otimes \mu $
unless $\mu$ is discrete.
Thus $\widetilde{P}$ does not generally correspond to a normal state 
on $L^\infty (\mu) \ntensor  L^\infty (\mu). $
Therefore the copying operations for the commutative algebras 
are not realizable if we define 
the compatibility relation by the normal tensor product.

Finally, to explicitly show the difference of the compatibility relations 
$\maxcomp$ and $\ncomp ,$
we show the following proposition.

\begin{prop}
\label{prop:dpovm}
Let $(\Omega , \Sigma , \mu)$ be a $\sigma$-finite measure space
and let $(\Omega , \Sigma , \oP)$ be the PVM on 
$\Hin = L^2 (\mu)$ defined by
$\oP (E) [f]_\mu := [\chi_E f]_\mu$
$(E \in \Sigma , [f]_\mu \in L^2 (\mu) ) .$
Then $\Gamma^{\oP} \ncomp \Gamma^{\oP}$ if and only if
$L^\infty (\mu)$ is atomic.
\end{prop}
\begin{proof}
We first show that $\oP$ is a maximal POVM
on $L^2 (\mu) .$
The QC channel $\Gamma^\oP$
is, up to normal isomorphism, 
the faithful representation in $\nchset{L^\infty (\mu)}{\cL (L^2 (\mu) )}$
given by
$
\Gamma^{\oP}([f]_\mu) [g]_\mu
:=
[fg]_\mu
$
$
(
[f]_\mu \in L^\infty (\mu) ,
\,
[g]_\mu \in L^2 (\mu)
) ,
$
from which we regard $L^\infty (\mu) \subseteq \cL (L^2 (\mu)) .$
Since $(L^2 (\mu) , \Gamma^\oP , \1_{L^2 (\mu)})$
is a minimal Stinespring representation of $\Gamma^\oP$
and $L^\infty(\mu)^\prime = L^\infty (\mu) ,$
Theorem~\ref{theo:maximalpovm} implies that
$\oP$ is a maximal POVM on $L^2 (\mu).$

Assume that $L^\infty (\mu)$ is atomic
and take a sequence of orthogonal projections
$\{ P_n\}_{n \in \natn} \subseteq L^\infty (\mu)$
such that
$
L^\infty (\mu) 
=
\Set{\sum_{n\in \natn} a_n P_n | a_n \in \cmplx ,\, \sup_{n \in \natn} |a_n| < \infty}.
$
Then $\Gamma^\oP$ is normally concatenation equivalent to the fully quantum channel
$\Gamma \in \nchset{ \cL ( \ell^2 (\natn) )}{\cL  ( L^2 (\mu)  )}$
defined by
$
\Gamma (A)
:=
\sum_{n\in \natn}
\braket{\delta_n | A \delta_n}
P_n ,
$
where $\ell^2 (\natn)$ is the Hilbert space of square-summable 
complex-valued functions on $\natn$
and
\[
	\delta_n (m):=
	\begin{cases}
	1, & (m=n);
	\\
	0, & (m \neq n).
	\end{cases}
\] 
Therefore from Corollary~\ref{coro:fq} we have 
$\Gamma^\oP \ncomp  (\Gamma^\oP)^c ,$
which implies $\Gamma^\oP \ncomp \Gamma^\oP$
since $\Gamma^\oP \concp ( \Gamma^\oP)^c .$

Assume $\Gamma^\oP \ncomp \Gamma^\oP .$
Since $\Gamma^\oP \eqncp (\Gamma^\oP)^c$
from the maximality of $\oP ,$
we have $\Gamma^\oP \ncomp  (\Gamma^\oP)^c .$
Thus Corollary~\ref{coro:fq} implies that
$\Gamma^{\oP}$ is normally concatenation equivalent to a 
fully quantum channel.
Furthermore, $\Gamma^\oP$ is minimal sufficient since
$\Gamma^\oP$ is injective.
Therefore, by applying a similar discussion in Theorem~5 of 
Ref.~\onlinecite{kuramochi2017minimal},
we conclude that
$L^\infty(\mu)$ is atomic.
\end{proof}

From Proposition~\ref{prop:dpovm}, we can conclude that
$\Gamma^\oP$ is not normally compatible with itself
if the space $L^\infty (\mu)$ is not atomic.
An example of non-atomic $ \mu$ is 
the Lebesgue measure on the real line 
$(\realn , \B (\realn)) .$
In this case, the PVM $\oP$ corresponds to the position measurement 
of a $1$-dimensional quantum particle.
Thus we have shown that the normal compatibility relation 
$\ncomp$ is strictly stronger
than the \cstar-compatibility relation $\maxcomp $
even for QC channels.

\section{Concluding remarks} \label{sec:conclusion}

In this paper, we have developed the theory of 
quantum incompatibility for channels with outcome operator algebras
and found that the concept of the commutant conjugate channel
plays a crucial role for the $\max$- and $\mathrm{bin}$-compatibility relations
$\maxcomp $ and $\bcomp .$
Still many problems remain unsolved,
some of which we list in the following.
\begin{enumerate}
\item
We have not obtained any non-trivial result for
the $\min$-compatibility relation $\mincomp .$
For example, the characterization of the class of $\min$-compatible
channels for a given channel remains to be unsolved.
\item
The operational meanings of the compatibility relations
$\maxcomp$ and $\mincomp$ for 
channels with non-commutative outcome spaces 
are also unclear.
\item
The definition of the commutant conjugate channel
strongly depends on the assumption 
that the input space is the full operator algebra 
$\LHin .$
For general \emph{input} von Neumann algebra $\Min$,
can we still have any generalization of the conjugate channel?
We can also ask how the class of compatible channels is characterized
in this case.
Note that the case of commutative $\Min$ is trivial because
the identity channel $\id_{\Min}$ is compatible with itself
in the sense of any \cstar-tensor product.
(The joint channel in this case is the representation $\Lambda$ 
given in Proposition~\ref{prop:ccomm}).
\item
Let $\mathfrak{C}^\sigma_\Lambda$ denote the class of 
normal channels normally compatible with a normal channel $\Lambda .$
For normal channels $\Lambda $ and $\Gamma ,$
we may ask whether the conditions 
$\Lambda \concp \Gamma$
and 
$\mathfrak{C}^\sigma_\Gamma \subseteq \mathfrak{C}^\sigma_\Lambda $
are equivalent.
If it is true for any QC channels,
this would be considered as another generalization of
the qualitative information-disturbance relation
proved for discrete POVMs in Ref.~\onlinecite{PhysRevA.88.042117}.
\end{enumerate}

\begin{acknowledgments}
The author would like to thank Teiko Heinosaari and Takayuki Miyadera
for helpful discussions and for sharing their research notes.
He also would like to thank Erkka Haapasalo for comments on the 
first version of this paper.
This work was supported by the National Natural Science Foundation of China 
(Grants No.\ 11374375 and No.\ 11574405).
\end{acknowledgments}

\end{document}